\documentclass{article}
\usepackage[utf8]{inputenc}
\usepackage{amssymb}
\usepackage{amsmath}
\usepackage{xparse}
\usepackage{graphicx}
\usepackage{amsthm}
\usepackage{comment}
\usepackage{float}
\usepackage{mathrsfs}
\usepackage{subcaption}
\usepackage{geometry}

\newtheorem{theorem}{Theorem}
\newtheorem*{theorem*}{Theorem}
\newtheorem{lemma}{Lemma}
\newtheorem{proposition}{Proposition}
\NewDocumentCommand{\evalat}{sO{\big}mm}{%
  \IfBooleanTF{#1}
   {\mleft. #3 \mright|_{#4}}
   {#3#2|_{#4}}%
}
\begin{document}
\title{One-dimensional hydrogenic ions\\ with screened nuclear Coulomb field}
\author{Suchindram Dasgupta, Chirag Khurana, A. Shadi Tahvildar-Zadeh\\
Rutgers (New Brunswick)}
\date{\today}
\maketitle
\begin{abstract}
We study the spectrum of the Dirac hamiltonian in one space dimension for a single electron in the electrostatic potential of a point nucleus, in the Born-Oppenheimer approximation where the  nucleus is assumed fixed at the origin.  The potential is screened at large distances so that it goes to zero exponentially at spatial infinity.  We show that the hamiltonian is essentially self-adjoint, the essential spectrum has the usual gap $(-mc^2,mc^2)$ in it, and that there are only finitely many eigenvalues in that gap, corresponding to ground and excited states for the system.  We find a one-to-one correspondence between the eigenfunctions of this hamiltonian and the heteroclinic saddle-saddle connectors of a certain dynamical system on a finite cylinder.  We use this correspondence to study how the number of bound states changes with the nuclear charge.
\end{abstract}

\section{Introduction and statement of main result}

Let $\phi = \phi(s)$ be the electrostatic potential due to a (point) nucleus in one-dimensional space.  Let $(x^0=t,x^1=s)$ be time and space coordinates on the 1+1-dimensional Minkowski spacetime.  Using the Born-Oppenheimer approximation, we can assume the nucleus is fixed at $s=0$.  Let $\Psi = \Psi(t,s) \in \mathbb{C}^2$ denote the wave function of a single electron placed in this electrostatic field.  According to the principles of relativistic quantum mechanics, $\Psi$ solves the one-dimensional 
Dirac equation with a minimal coupling to the potential $\phi$ (we have set $\hbar = c=1$):
\begin{equation}
     -i\gamma^{\mu}\partial_{\mu}\Psi-e \gamma^0 \phi(s) \Psi + m\Psi=0.
\end{equation}
Here $\gamma^\mu$ are $2\times 2$ matrices satisfying $\gamma^\mu\gamma^\nu + \gamma^\nu \gamma^\mu = 2 \eta^{\mu\nu} I_{2\times 2}$ for $\mu,\nu\in\{0,1\}$,  $\eta := \mbox{diag}(1,-1)$ is the Minkowski metric, $\partial_\mu := \frac{\partial}{\partial x^\mu}$, and Einstein summation convention is used. $e>0$ is the elementary charge (i.e. $e$ is the charge of a proton and $-e$ is the charge of an electron), and $m$ is the mass of the electron. 

Writing the above in hamiltonian form, we obtain
\begin{equation}\label{Direq}
  i\partial_{t}\Psi = -i\gamma^{0}\gamma^{1}\partial_{s}\Psi - e \phi \Psi + m\gamma^0\Psi =: H_D \Psi.
\end{equation}
Letting \(\alpha^1 := \gamma^0\gamma^1\) and \(\beta := \gamma^0\), we get
\begin{equation}\label{DirHam}
    H_D = - i \alpha^1 \partial_{s} - e \phi  I + m \beta.
\end{equation}
We may choose the following representation
\begin{equation}
    \gamma^0=\begin{pmatrix} 0 & 1 \\ 1 & 0  \end{pmatrix},\qquad
\gamma^1=\begin{pmatrix} 0 & -1 \\ 1 & 0  \end{pmatrix},
\end{equation}
which yields
\begin{equation}\label{eq:Ham}
    H_D= \begin{pmatrix} -e \phi - i\partial_{s} & m \\ m & -e \phi + i\partial_{s}  \end{pmatrix}.
\end{equation}
To study the spectrum of \(H_{D}\), we search for solutions of \eqref{Direq} that are of the form 
\begin{equation}
    \Psi(t,s) = e^{-iEt}\psi(s),
\end{equation}
which leads us to the eigenvalue problem
\begin{equation} \label{EVP}
H_D \psi(s) = E \psi(s).
\end{equation}
A nonzero $\psi$ that satisfies \eqref{EVP} and is square-integrable, i.e. $\int_{\mathbb{R}} |\psi|^2 ds < \infty$, is called an eigenfunction for $H_D$, and in that case the corresponding number $E$ is called an energy eigenvalue.  

In this paper we prove that, for a choice of $\phi$ that corresponds to a screened nuclear Coulomb field, i.e. 
\begin{equation}\label{screened}
\phi(s) = \frac{Ze}{2} \exp\{-|s|\},
\end{equation} 
with $Z$ the number of protons in the nucleus, the above eigenvalue problem always has a finite number of solutions. The following is an informal statement of our main result.  For the precise statement, see Theorem~\ref{Theorem:main}.
\begin{theorem*} (Informal statement) For any value of the parameter $\gamma := Ze^2>0$, the eigenvalue problem \eqref{EVP} for the hamiltonian \eqref{eq:Ham} with electrostatic potential $\phi$ as in \eqref{screened}, has a finite number $N = N(\gamma)$ of solutions, consisting of a ground state and $N-1$ excited states, with energy eigenvalues $E_0< E_1 <  \dots < E_{N-1}$ lying in the interval $(-m,m)$. The corresponding eigenfunctions $\psi_0,\dots \psi_{N-1}$ can be put in one-to-one correspondence with the heteroclinic saddle-saddle connectors $W_0,W_1,\dots,W_{N-1}$ of a certain dynamical system on a finite cylinder, in such a way that each $W_j$ has a well-defined winding number around the cylinder that increases with $j$.  The function $N(\gamma)$ is an integer-valued function whose points of discontinuity are precisely the zeros and critical points of certain Whittaker functions.
\end{theorem*}

The rest of this paper is organized as follows: In Section~\ref{sec:redham} we derive the reduced hamiltonian, establish its self-adjointness, and set up the dynamical system used to study its discrete spectrum. In Section~\ref{sec:existence} we establish the existence of energy eigenvalues corresponding to ground and excited states of the system. Section~\ref{Numerical Section} is devoted to numerical investigations we conducted on the equations.  We conclude with a summary and outlook in Section~\ref{sec:sumout}.  Various technical results on the behavior of the special functions used in the proofs of our theorems are gathered together in the Appendix at the end of the paper.

\section{The reduced hamiltonian}\label{sec:redham}
\subsection{Derivation}
Referring back to equation \eqref{EVP}, we need to find both the eigenfunction $\psi$ and the energy eigenvalue $E$. The wave function has two complex-valued components. Let us set
\begin{gather}
    \psi = 
    \begin{pmatrix} R_1(s)\\R_2(s) \end{pmatrix},
\end{gather}
where $R_1$ and $R_2$ are complex-valued functions of one real variable. Plugging this back into equation \eqref{EVP}, we get
\begin{gather}
    \begin{pmatrix} -e \phi - i\frac{d}{ds} & m \\ m & -e \phi + i\frac{d}{ds}  \end{pmatrix}\begin{pmatrix} R_1(s)\\R_2(s) \end{pmatrix} = E\begin{pmatrix} R_1(s)\\R_2(s) \end{pmatrix},
\end{gather}
which yields the following set of ordinary differential equations:
\begin{gather}
    R_1'-i(E+e\phi)R_1+imR_2 = 0 \label{eq:R1},\\
    -R_2'-i(E+e\phi)R_2+imR_1 = 0 \label{eq:R2}.
\end{gather}
These are equations for two unknown complex-valued, and therefore four unknown real-valued functions.  We proceed to simplify this system by reducing the number of independent real unknown functions to two.  Multiplying \eqref{eq:R1} with $R_1^*$ and \eqref{eq:R2} with $R_2^*$, adding the two resulting equations and taking the real part, we obtain
\begin{equation}
   \frac{d}{ds}(|R_1(s)|^2-|R_2(s)|^2) = 0,
\end{equation}
which implies that $\Delta := |R_1(s)|^2-|R_2(s)|^2$ is a constant, i.e. independent of $s$.  But if $\psi$ is an eigenfunction, it must be square-integrable.  Thus $|R_1|^2$ and $|R_2|^2$ are both integrable over $\mathbb{R}$, hence so is their difference, which implies that $\Delta  = 0$.  Thus 
\begin{gather*}
    |R_1(s)|=|R_2(s)|=:R(s)
\end{gather*}
for all $s$.  This allows us to write the wave function components in the following way:
\begin{gather}\label{phases}
    R_j(s) = R(s)e^{i\varphi_j},\quad j=1,2.
\end{gather}
We now want to show that  without loss of generality  $\varphi_1+\varphi_2=0$; that is
\begin{gather*}
R_2 = R(s)e^{-i\varphi_1}=R_1^*.
\end{gather*}
First, we take equation \eqref{eq:R1} and multiply it by the conjugate of $R_2$, then take the conjugate of equation \eqref{eq:R2} and multiply it by $R_1$:
\begin{gather}
    R_2^*(R_1'-i(E-e\phi)R_1+imR_2) = 0,\\
    R_1(-R_2'-i(E-e\phi)R_2+imR_1)^* = 0.
\end{gather}
Now we add these equations to obtain
\begin{gather}
    R_1'R_2^*-R_2'^*R_1 = 0,
\end{gather}
which implies that
\begin{gather}
    (R_2^*)^2\frac{d}{ds}\left(\frac{R_1}{R_2^*}\right)= 0.
\end{gather}
If $R_2^*(s) = 0$ for some $s$, then $R_1(s) = 0$ as well (since they have the same magnitude, $R(s)$,) and thus their phase is irrelevant. Let $Z = \{s\in\mathbb{R}\ |\ R(s) = 0\}$.  $Z$ is closed, so its complement in $\mathbb{R}$ is a union of open intervals. On any of those intervals, 
from \eqref{phases} we see that
\begin{gather}
    \frac{d}{ds}(\varphi_1+\varphi_2)=0,
\end{gather}
which means that
\begin{gather}\label{varphis}
    (\varphi_1+\varphi_2)=\delta,
\end{gather}
where \(\delta\) is some constant.
Recall that
\begin{equation}
    \psi= \begin{pmatrix} R_1\\R_2 \end{pmatrix} = \begin{pmatrix}  Re^{i\varphi_1}\\Re^{i\varphi_2} \end{pmatrix}.
\end{equation}
If $\psi$ is a solution of \eqref{EVP}, multiplying it by a constant phase factor \(e^{i\theta}\) will still be a solution. Choosing the phase factor to be \(e^{-i\delta/2}\), we find that
\begin{equation}
    \psi' = \begin{pmatrix}  Re^{i(\varphi_1-\delta/2)}\\Re^{i(\varphi_2-\delta/2)} \end{pmatrix}
\end{equation}
is equivalent to $\psi$.  Let \(\varphi_1'=\varphi_1-\delta/2\) and \(\varphi_2'=\varphi_2-\delta/2\). Then by \eqref{varphis}
we have $\varphi_1' + \varphi_2' = 0$.
Therefore, without  loss of generality \(\varphi_1+\varphi_2\) can be set equal to 0. ie. \( \varphi_1 = -\varphi_2\). In that case \(R_1\) and \(R_2\) are complex conjugates of one another. Therefore, we can set
\begin{equation}
    \left\{ \begin{aligned} 
 R_1 &= \frac{1}{\sqrt{2}}(u - iv),\\
  R_2 &= \frac{1}{\sqrt{2}}(u + iv).
\end{aligned} \right.
\end{equation}
Remembering that we are only interested in square integrable solution of \eqref{EVP}, we must have $\int_{-\infty}^{\infty} (|R_1|^2 + |R_2|^2) \, ds < \infty$, so that $\psi$ can be normalized in such a way that this quantity is one. This now implies
\begin{equation}
 \int_{-\infty}^{\infty} (u^2 + v^2) \,ds  = 1.
\end{equation}
Recall from the beginning of the section that
\begin{gather}
    R_1'-i(E+e\phi)R_1+imR_2 = 0, \\
    -R_2'-i(E+e\phi)R_2+imR_1 = 0.
\end{gather}
This becomes, 
\begin{equation}
    \left\{ \begin{aligned} 
 u' + (-m - e \phi) v - Ev  &= 0,\\
 -v' + (m-e\phi)u-Eu &= 0.
\end{aligned} \right.
\end{equation}
The above can be rewritten as 
\begin{gather}
h \begin{pmatrix}
    u \\ v
    \end{pmatrix} = E \begin{pmatrix}
    u \\ v
    \end{pmatrix},
    \label{eqn:reduced}
\end{gather}
where \(h\) is our \textbf{reduced hamiltonian}:
\begin{equation}\label{def:h}
     h: = \begin{pmatrix}
m-e\phi & -\frac{d}{ds} \\ \frac{d}{ds} & - m - e\phi
\end{pmatrix}.
\end{equation}
\subsection{The spectrum of the reduced hamiltonian}
Earlier, we obtained the system of equations and constraint
\begin{gather}\label{eq:eig}
(h-E) \begin{pmatrix}
    u \\ v
    \end{pmatrix} = \begin{pmatrix}
    0 \\ 0
    \end{pmatrix},\\
    \int_{-\infty}^{\infty} (u^2 + v^2) \,ds  = 1,
\end{gather}
where the reduced hamiltonian $h$ is as in \eqref{def:h}. The hamiltonian can be written in the form
\begin{gather}\label{weidform}
    h= J\frac{d}{ds} + P,
\end{gather}
where 
\begin{gather}
  J := \begin{pmatrix}
0 & -1 \\ 1 & 0
\end{pmatrix} \qquad P(s) := \begin{pmatrix}
    m-e\phi(s) & 0 \\ 0 & -m-e\phi(s)
    \end{pmatrix}.
\end{gather}
In order to study hamiltonians like this further, we need some preliminaries: We need a Hilbert space $\mathcal{H}$, i.e. a vector space over $\mathbb{C}$ of pairs of functions $(u,v)$ on which the  hamiltonian can act, together with a complex innerproduct $\langle \ , \ \rangle$ defined on it, and we need $\mathcal{H}$ to be complete with respect to the norm given by this innerproduct.  In our case we take $\mathcal{H} = (L^2(\mathbb{R}))^2$ i.e. the set of pairs of square-integrable functions defined on the real line, together with the standard $L^2$-innerproduct $\langle f,g\rangle = \int f_1 g_1^*+f_2 g_2^*\ dx$. The operator $h$ needs to be defined on a linear subspace $\mathcal{D}(h)$ of $ \mathcal{H}$, called the {\em domain} of $h$, and we need this domain to be dense in $\mathcal{H}$. Since \eqref{eq:eig} formally looks like an eigenvalue-eigenvector equation, and we expect $E$ to be real, we need \(h\) to be {\em self-adjoint}.  Recall that in order for a matrix of numbers to have real eigenvalues, it must be hermitian-symmetric, i.e. equal to its own conjugate-transpose.  For an {\em operator-valued} matrix such as $h$ this is not enough, and more care is needed in order to determine its self-adjointness.  (See e.g. \cite{Reed1975MethodsOM} Vol. 2.)  In particular, for an operator to be self-adjoint, in addition to being symmetric with respect to the given innerproduct, it is necessary that its domain and the domain of its {\em adjoint} coincide (see e.g. Teschl \cite{TeschlBOOK} for the definition of adjoint and the criteria for self-adjointness.)

We can then ask, what is the \textbf{spectrum} $\sigma(h)$ of \(h\)? By spectrum, we mean all \(\lambda \in \mathbb{C}\) that make the operator \(h-\lambda I\) not have a bounded inverse. For a self-adjoint $h$ the set $\sigma(h)$ is the union of two disjoint subsets 
$$ \sigma(h) = \sigma_{disc} \cup \sigma_{ess}.
$$
The discrete spectrum $\sigma_{disc}$ consists of numbers satisfying the traditional notion of eigenvalue, i.e. isolated points $\lambda \in \mathbb{C}$ such that the nullspace $\mathcal{N}(h-\lambda I) \subset \mathcal{D}(h)$ is non-trivial and finite dimensional. Anything else in $\sigma(h)$ belongs to the essential spectrum $\sigma_{ess}(h)$.  This includes accumulation points of eigenvalues and eigenvalues with infinite-dimensional eigenspace, as well as the {\em continuous spectrum}, where the only candidates $\psi$ for satisfying the eigenvalue equation $h \psi = \lambda \psi$ do not belong to the Hilbert space $\mathcal{H}$ (see \cite{TeschlBOOK} for precise definitions and statements.)

In our case, it can be shown that the discrete spectrum of the Dirac hamiltonian corresponds to the bound states of the electron with the nucleus, while the essential spectra 
correspond to the scattering states of this system (see e.g. Thaller \cite{ThallerBOOK}.)

In 3 dimensions, the spectrum of the Dirac operator with the Coulomb potential $H_{DC}$ is the following
\begin{gather}
    \sigma_{ess}(H_{DC})= (-\infty,m] \cup [m,\infty) \\
    \sigma_{disc}(H_{DC})= \left\{E_{n}\right\}_{n=1}^{\infty} \subset (-m,m).
\end{gather}
where \(m\) is the mass of the electron. For the discrete spectrum,
\begin{equation}
    E_0 < E_1 < E_2 < \ldots 
\end{equation}
which correspond to the orbital energies of hydrogen, with \(E_0\) being the ground state energy. Additionally, \(E_n \rightarrow m\) as \(n \rightarrow \infty\). 
We wish to replicate all of these properties in 1 dimension.  For the above results to hold, it is necessary that
\begin{equation}\label{condphi}
    \phi(s) \rightarrow 0 \qquad \mbox{ as }|s| \rightarrow \infty . 
\end{equation}
However, in one space dimension the electrostatic potential of a point charge placed at $s=0$ satisfies $-\phi''(s) = Q\delta_0(s)$ where $\delta_0$ is the Dirac delta distribution.  Therefore \(\phi = - \frac{Q}{2}|s|\), which does not go to zero at infinity, so \eqref{condphi} fails. In what follows we will replace $\phi$ with a screened version of itself, one that has the same absolute value behavior at the origin but decays exponentially fast at infinity.

Consider the potential 
\begin{gather}
    \phi(s) = \frac{Q}{2}\mu e^{-|s|/\mu},
\end{gather}
where $\mu$ is a screening length (which for now we will set equal to 1.)
\begin{proposition}\label{prop:Widemann}
With the above $\phi$, the reduced hamiltonian h is self-adjoint, and its essential spectrum is $(-\infty,-m]\cup [m,\infty)$. Its discrete spectrum, if non-empty, will consist only of simple eigenvalues (i.e. the eigenspaces will be one-dimensional.)
\end{proposition}
\begin{proof}
Recall that $h = J\frac{d}{ds} + P$, and since $\psi(s) \to 0$ as $|s|\to \infty$, we have
\begin{gather}\label{cond:weid}
    P(s) \rightarrow \begin{pmatrix}
    m & 0 \\ 0 & -m
    \end{pmatrix},\qquad\mbox{ as }|s|\to \infty
\end{gather}
The conclusions about self-adjointness, $\sigma_{ess}$, and $\sigma_{disc}$ all follow from Weidmann's \cite{weidmann2006spectral} Theorems 16.5,  16.6, and 10.8, respectively, about one-dimensional hamiltonians of the form \eqref{weidform} that satisfy \eqref{cond:weid}.
\end{proof}
Recall that the {\em coupled} system of linear ordinary differential equations \eqref{eq:eig} can be written more explicitly as:
\begin{equation}\label{eq:uv}
    \left\{ \begin{aligned} 
 \frac{du}{ds} + (-m - e \phi) v - Ev  &= 0,\\
 -\frac{dv}{ds} + (m-e\phi)u-Eu &= 0.
\end{aligned} \right.
\end{equation}
The coupling of the two differential equations for \(u\) and \(v\), and the fact that $E$ is also unknown, makes this problem a bit complicated. The problem might be simplified, however, if the equations could be decoupled. One way of doing this is through a {\em Pr\"ufer transform}, as was done in \cite{KiTa15}  (See \cite{UlHo86} for an earlier application of this method):
Let us define
\begin{equation}
    R^2 := u^2 + v^2,\qquad
    \theta(s) := \arctan{\frac{v(s)}{u(s)}},
\end{equation}
so that 
$$ u = R\cos \theta,\qquad v = R \sin\theta.$$
Then we have
\begin{gather}
    R' = \frac{1}{R} (uu'+ vv') = \frac{1}{R}[u((m+e\phi)v+Ev)+v((m-e\phi)u-Eu)] = \frac{2m}{R} uv,
    \end{gather}
    so that
    \begin{gather}
      \frac{R'}{R}= m \sin{2\theta}.
\end{gather}
We can therefore solve for \(R\) if \(\theta\) is known. Similarly, we have
\begin{gather}
    \theta' = \frac{1}{1+\frac{v^2}{u^2}}\frac{v'u-u'v}{u^2}
     = \frac{v'u-u'v}{R^2}.
\end{gather}
Substituting from \eqref{eq:uv},
\begin{gather}
    \theta' = \frac{(m-e\phi-E)u^2 - (m+e\phi+E)v^2}{R^2}
     = m\cos(2\theta) - e\phi - E.
\end{gather}
We now have a new system of \textit{partially decoupled} differential equations ($\theta$ equation has no $R$):
\begin{equation}
 \frac{R'}{R}= m \sin{2\theta}, \qquad
 \theta' = m\cos(2\theta) - e\phi(s) - E.
\end{equation}
This system must satisfy the condition that 
\begin{equation}
   \int_{-\infty}^{\infty} R^2 \,ds  = 1.
    \label{eqn:normalize}
\end{equation}
Since the $\theta$ equation does not involve $R$, we can focus on analyzing the $\theta$ equation first.  We do this by converting the $\theta$ equation into a dynamical system on a 2-dimensional surface.
\subsection{Setting up a dynamical system}
We first make the $\theta$ equation \textit{autonomous}. This means that we do not want the independent variable to show up on the right side of the differential equation. This can be done trivially by introducing a new independent variable $\tau$ and setting $s(\tau) = \tau$. Then
\begin{equation}\label{dynsys}
   \left\{\begin{array}{rcl}\dot{s} &=& 1\\ \dot{\Theta} &=& 2m\cos(\Theta) - 2e\phi(s) - 2E,
    \end{array}\right.
\end{equation}
where dot denotes differentiation with respect to $\tau$ and we have set $\Theta = 2\theta$ for simplification. We now have $\tau$ as an independent variable and $\Theta$ and $s$ as dependent variables. 

Next we recall the equation for $R$:
\begin{gather}\label{eq:R}
    \frac{R'}{R} = m\sin(\Theta).
\end{gather}
Solving this equation, we get
\begin{gather}
    R(s) = R(0)\exp\left\{\int_0^s m\sin(\Theta(s)) \,ds\right\}.
    \label{eqn:R_theta}
\end{gather}
Since $R\in L^2$ and $\sin\Theta$ is bounded, from \eqref{eq:R} we have that $R'\in L^2$ as well. It turns out that an $L^2$ function whose derivative is also $L^2$ must go to zero at infinity:  $R(s) \to 0$ as $|s| \to \infty$ (see e.g. \cite{RB78}.)
Thus the integral in the exponent in \eqref{eqn:R_theta} must diverge to $-\infty$. So $\sin\Theta$ must be negative as $s \to \infty$, but positive as $s\to -\infty$, i.e.:
\begin{gather}
    \Theta(\infty) \in [-\pi,0),\qquad
    \Theta(-\infty) \in (0,\pi].
\end{gather}
Next, we want to compactify the system \eqref{dynsys}.  To this end let us now define a new variable
\begin{gather}
    z = \arctan(s),
\end{gather}
so that $s=\pm\infty \Longleftrightarrow z = \pm\frac{\pi}{2}$.
Changing variables in \eqref{dynsys}, we obtain
\begin{equation}
    \left\{\begin{array}{rcl}\dot{z} &=& \cos^2{z},\\
    \dot{\Theta} &=& 2m\cos(\Theta)-2e\phi(\tan{z})-2E.
    \end{array}\right.
\end{equation}
Recall that
\begin{gather}
    \phi(s) = \frac{Q}{2}\exp(-|s|).
\end{gather}
We can also choose units such that $m=1$. So, in the case of a nucleus with $Z$ protons fixed at the origin, $Q = Ze$ and the system becomes 
\begin{gather}\label{zThetasys}
    \left\{\begin{array}{rcl}\dot{z} &=& \cos^2{z} =: F(z,\Theta),\\
    \dot{\Theta} &=& 2\cos(\Theta)-\gamma\exp{(-|\tan{z}|)}-2E =: G_E(z,\Theta),
    \end{array}\right.
\end{gather}
where 
$ \gamma := Ze^2$.
\subsection{Linearizing the System}
One way we can study this system's behavior further is through a local linear approximation near its equilibrium points. The local linear approximation of \eqref{zThetasys} about an equilibrium point \( (z_0,\Theta_0) \) would be
\begin{gather}
    \frac{d}{d\tau} \begin{pmatrix}
    z-z_0 \\ \Theta - \Theta_0
    \end{pmatrix} = \boldsymbol{J}(z_0,\Theta_0) \begin{pmatrix}
    z-z_0 \\ \Theta-\Theta_0
    \end{pmatrix},
\end{gather}
where the \(2\times 2\) matrix \(\boldsymbol{J}(z_0,\Theta_0)= \begin{pmatrix}
    F_z(z_0,\Theta_0) && F_\Theta(z_0,\Theta_0) \\ {G_E}_z(z_0,\Theta_0) && {G_E}_\Theta(z_0,\Theta_0) 
    \end{pmatrix} \) is the Jacobian matrix. We can now compute the partial derivatives as follows. 
\begin{equation}
 F_z = -2 \cos{z}\sin{z},\quad
 F_\Theta = 0,\quad
 {G_E}_z = 2e \phi' (\tan{z}) \sec^2{z},\quad
 {G_E}_\Theta = -2\sin{\Theta}.
\end{equation}
Recall \(\phi\) is Lipschitz at the origin and smooth otherwise. \(F(z,\Theta)\) is \(0\) at \(z= \pm \frac{\pi}{2}\). So, our equilibrium points lie on either \(z = -\frac{\pi}{2}\) or \(z=\frac{\pi}{2}\). When substituting either value into \(G_E(z,\Theta)\), we obtain:
\begin{gather}
\Theta = \pm \cos^{-1}{E},
\end{gather}
where by $\cos^{-1}$ we mean the branch of $\arccos$ with values in the interval $[0,\pi]$. If \(|E| < 1\), then there are 4 equilibrium points (modulo $2\pi$), which are as follows.
\begin{align}
    S_E^{-}&:(-\frac{\pi}{2},\cos^{-1}{E}), & S_E^{+}&:(\frac{\pi}{2},-\cos^{-1}{E}), \\ N_E^{-}&:(-\frac{\pi}{2},-\cos^{-1}{E}), & N_E^{+}&:(\frac{\pi}{2},\cos^{-1}{E}).
\end{align}
In this case, the Jacobian at our equilibrium points is:
\begin{gather}
    \boldsymbol{J}(\pm \frac{\pi}{2}, \cos^{-1}{E}) = \begin{pmatrix}
    0 && 0 \\ 0 && -2\sqrt{1-E^2}
    \end{pmatrix}, \\ 
    \boldsymbol{J}(\pm \frac{\pi}{2}, -\cos^{-1}{E}) = \begin{pmatrix}
    0 && 0 \\ 0 && 2\sqrt{1-E^2}
    \end{pmatrix}.
\end{gather}
If \(|E| = 1\), then there are only 2 equilibrium points modulo $2\pi$, with Jacobians \(\boldsymbol{J}=\begin{pmatrix}
    0 && 0 \\ 0 && 0
    \end{pmatrix} \), so these equilibrium points are completely degenerate.  For $E=-1$ we have the equilibrium points
    \begin{equation}
        C^- := (-\frac{\pi}{2},\pi),\qquad C^+ := (\frac{\pi}{2},\pi),
    \end{equation}
    while for $E=1$ we have
\begin{equation}
    D^{-}:=(-\frac{\pi}{2},0)\qquad D^{+}:=(\frac{\pi}{2},0).
\end{equation}
If \(|E| > 1\), then there are no equilibrium points.

Since we have found the linearization of our system at the equilibrium points in each case, we can also find  the eigenvalues and associated eigenvectors of the locally linear system. For $|E|<1$ this will give us information about the behavior of the system near the equilibrium point. 
The eigenvalues of 
\(
    \boldsymbol{J}(\pm \frac{\pi}{2}, \cos^{-1}{E})
\)
are \(\lambda_1 = 0  \) and \(\lambda_2 = -2\sqrt{1-E^2}  \) with corresponding eigenvectors being \(\begin{pmatrix}
1 \\ 0
\end{pmatrix}\) and \(\begin{pmatrix}
0 \\ 1
\end{pmatrix}\) 
respectively. Similarly, the eigenvalues of \(
    \boldsymbol{J}(\pm \frac{\pi}{2}, -\cos^{-1}{E})
\) are \(\lambda_1 = 0  \) and \(\lambda_2 = 2\sqrt{1-E^2}  \) with corresponding eigenvectors being \(\begin{pmatrix}
1 \\ 0
\end{pmatrix}\) and \(\begin{pmatrix}
0 \\ 1
\end{pmatrix}\) 
respectively.

The purpose of the linearization was to determine the behavior of the trajectories near the equilibrium points in the phase portrait of our system. This is complicated by the fact that the equilibria of this system are {\em non-hyperbolic}, meaning there are zero eigenvalues.  This means that we need {\em center manifold theory} (see e.g. Carr \cite{CarrBOOK}) to describe the behavior of the nonlinear system.

According to this theory, when $|E|<1$, the equilibrium points $S_E^\pm$ and $N_E^\pm$ correspond to saddle-nodes. Their local behavior is
determined by Theorem 2.19(iii) in \cite{Dumortier2006QualitativeTO}. Their local phase portraits are depicted in Figure 2.13(c)
of the same reference (see also Figure~\ref{fig:my_label}.) The uniqueness of the center-unstable manifold emanating from $S_E^-$ follows from this
theorem. Similarly for the center-stable manifold going into $S_E^+$. For a generic value of $E$, these two orbits will not coincide, i.e., generically, the orbit from $S_E^-$ will run into $N_E^+$, and the orbit that goes into $S_E^+$, when run backwards, will fall into $N_E^-$.  

Additionally, recall that we need $\sin\Theta$ to be negative at \(s=\infty\) (i.e. \(z=\pi/2\)) and positive at \(s=-\infty\) (ie. \(z=-\pi/2\)). The equilibrium points which correspond to these conditions are \(S_E^+\) and \(S_E^-\) respectively. Therefore, the energy of a bound state for the electron in our system will be the energy level that gives  a trajectory between these two equilibrium points, i.e. the value for $E$ that makes the center-unstable manifold of $S_E^-$ coincide with the center-stable manifold of $S_E^+$, resulting in a heteroclinic orbit connecting these two saddle-nodes. See Fig.~\ref{fig:my_label}.
\begin{figure}
    \centering
    \includegraphics[scale=0.45]{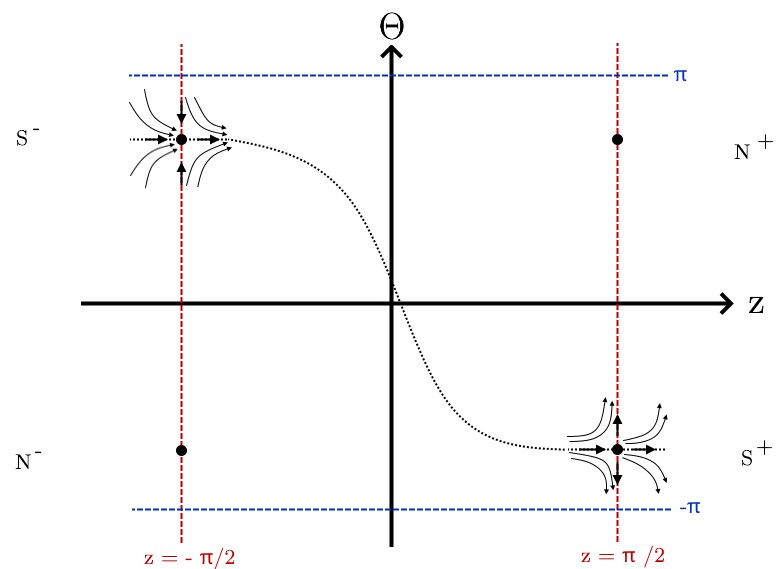}
    \caption{A heteroclinic orbit connecting two saddle-node equilibrium points. }
    \label{fig:my_label}
\end{figure}
Because the dynamical system \eqref{zThetasys} is $2\pi$-periodic in $\Theta$, one can view it as a dynamical system on a finite cylinder $[-\frac{\pi}{2},\frac{\pi}{2}]\times \mathbb{S}^1$.
As a result, there may exist connectors between \(S_E^-\) and \(S_E^+\) that start on the left boundary and wrap around the cylinder multiple times before reaching the right boundary of the cylinder. For the purpose of analyzing the system it is convenient to ``unwrap" the cylinder into a vertical strip $[-\frac{\pi}{2},\frac{\pi}{2}]\times \mathbb{R}$ (the universal cover of the cylinder) and identify the rectangle $[-\frac{\pi}{2}, \frac{\pi}{2}] \times [-\pi, \pi]$ as the {\em fundamental domain}.
Note that if we define a saddles connector as beginning in our fundamental domain at $S_E^-$, then it must connect either to the $S_E^+$ which is in the fundamental domain, or to \textit{some copy} of $S_E^+$ shifted down or up by some multiple of $2\pi$. The number of times a trajectory wraps around the cylinder is called the \textbf{winding number} of the orbit. We define this as
\begin{equation}
    N = \left\lfloor{\frac{\Theta(-\infty)-\Theta(\infty)}{2\pi}}\right\rfloor.
\end{equation}
Here $\left\lfloor x\right\rfloor$ denotes the largest integer less than or equal to $x$.
We would like to find out whether or not saddles connectors exist for any integer winding number. 
\section{Existence of energy eigenvalues and eigenfunctions}\label{sec:existence}

\subsection{Existence of saddles connectors with a given winding number}

Here we apply a continuity argument first given in \cite{KiTa15} and subsequently generalized in \cite{KLT22}, to prove the existence of heteroclinic saddles connectors with any given winding number $N$. 
We begin by recalling some standard terminology from dynamical systems:  For the autonomous system of differential equations 
\begin{equation}\label{eq:ds}
\dot{\mathbf{y}} = \mathbf{F}(\mathbf{y}),\qquad \mathbf{F}: \mathbb{R}^n \to \mathbb{R}^n
\end{equation}
with $\mathbf{F}\in C^1$, the {\em flow map} $\Phi_t:\mathbb{R}^n \to\mathbb{R}^n$ is defined as $\Phi_t(\mathbf{x}) = \mathbf{y}(t)$, where $\mathbf{y}:I\to \mathbb{R}^n$ is the unique solution to \eqref{eq:ds} with initial value $\mathbf{y}(0) = \mathbf{x}$ that is guaranteed to exist for some open interval $I$ around $t=0$. The {\em $\omega$-limit set} of a point $\mathbf{x}_0$ is then defined by
$$ \omega(\mathbf{x}_0) := \{ \mathbf{x}\in \mathbb{R}^n\ |\ \exists \mbox{ sequence } t_n \to \infty \mbox{ s.t. } \Phi_{t_n}(\mathbf{x}_0) \to \mathbf{x} \}.$$
The {\em $\alpha$-limit set} is defined analogously, with $\infty$ replaced by $-\infty$.  It is clear that all points on an orbit have the same $\alpha$- and $\omega$-limit sets, thus it makes sense to talk about $\alpha$- and $\omega$-limits of orbits.

\begin{theorem}\label{Theorem: GT}
Let \(N \in \mathbb{Z}\) be an integer and \(w_E\) be the energy-dependent winding number of trajectories whose $\alpha$-limit is \(S_E^-\). Then, for energy values \(-1 \leq E' < E'' \leq 1\) such that 
\begin{equation}
\label{assumpE}
w_{E'} \leq N \quad \textrm{and} \quad w_{E''} \geq N+1,
\end{equation}
there exists some \(E \in (E',E'')\) such that there is a saddles connector \(\mathcal{W}_E\) with winding number \(w_E = N\). 
\end{theorem}

\begin{proof}
Define the orbit \(\mathcal{W}_{E'}^-\)  to be one with energy \(E' \) whose $\alpha$-limit is $S_{E'}^-$ in our fundamental domain and $\omega$-limit located above a particular copy of $S_{E'}^+$ at $z=\pi/2$  (i.e., above \((\frac{\pi}{2},\arccos{(E')}-2\pi N)\)). Similarly  define \(\mathcal{W}_{E''}^-\) to be the orbit with a higher energy \(E''\) whose $\alpha$-limit is $S_{E''}^-$ in our fundamental domain and $\omega$-limit located at some point below the same copy of $S_{E''}^+$ at $z=\pi/2$ (i.e., below \((\frac{\pi}{2},\arccos{(E'')}-2\pi N)\)). The existence of these orbits is guaranteed by the assumption \eqref{assumpE}.

Additionally, define orbits \(\mathcal{W}_{E'}^+\) and \(\mathcal{W}_{E''}^+\) whose $\omega$-limits are \(S_{E'}^+\)
and \(S_{E''}^+\) in the fundamental domain of our phase portrait \(\mathcal{C^*}= [-\frac{\pi}{2},\frac{\pi}{2}] \times [-\pi,\pi]\) at the energy levels \(E'\) and \(E''\) respectively. If the $\alpha$-limit of one of these is $S_{E'}^-$ or $S_{E''}^-$, we already have a saddles connector and we're done, so we can assume that these orbits will run backward into some copy of \(N_{E'}^-\) and \(N_{E''}^-\), respectively. Lastly, define \(\sigma_{E'}\) and \(\sigma_{E''}\) to be the orbits that are equivalent to  \(\mathcal{W}_{E'}^+\) and \(\mathcal{W}_{E''}^+\), but whose $\omega$-limits are the copies of the corresponding \(S_{E'}^+\) and \(S_{E''}^+\) shifted down by \(2\pi N\), (i.e., the points \((\frac{\pi}{2},\arccos{(E')}-2\pi N)\) and \((\frac{\pi}{2},\arccos{(E'')}-2\pi N)\) respectively. 

Define \(\mathcal{K}_{E'}\) as the open domain on our cylinder such that \(\mathcal{W}_{E'}^-\) and \(\sigma_{E'}\) lie on the boundary \(\partial \mathcal{K}_{E'}\), and define \(\mathcal{K}_{E''}\) as the open domain on our cylinder such that \(\mathcal{W}_{E''}^-\) and \(\sigma_{E''}\) lie on the boundary \(\partial \mathcal{K}_{E''}\). Orient each boundary so that the orientation induced on \(\mathcal{W}_{E'}^-\) and \(\mathcal{W}_{E''}^-\) coincides with the direction of the flow (i.e., left to right). By Green's theorem, the signed area of \(\mathcal{K}_{E}\) is 
\[A(E) =  \oint_{\partial\mathcal{K}_{E}} -\Theta\, dz =  \int_{-\pi/2}^{\pi/2}(y_{E}^{+}-y_{E}^{-})\,dz,\]
where \(y_{E}^{-}\) denotes the \(\Theta\) component of \(\mathcal{W}_E^-\) and \(y_{E}^{+}\) denotes the \(\Theta\) component of \(\sigma_{E}\). Orbits in our dynamical system cannot intersect with one another, so either \(y_{E}^{+}-y_{E}^{-} \geq 0\) or \(y_{E}^{+}-y_{E}^{-} \leq 0\) for all \(z \in (-\pi/2, \pi/2)\). 

Therefore, if a value of \(E\) exists such that \(A(E)=0\), then \(y_{E}^{+}=y_{E}^{-}\) and the orbits coincide. The right-hand equilibrium point of \(\mathcal{W}_{E'}^-\) is \(N_{E'}^+\) in our fundamental domain while the right-hand equilibrium point of \(\mathcal{W}_{E''}^-\) is below our fundamental domain. This implies that $A(E')<0<A(E'')$.  If $A$ is a continuous function of 
\(E\) in the interval \( [E',E'']\), then by the Intermediate Value Theorem, there exists some \(E \in (E',E'')\) such that \(A(E)=0\), implying the existence of a saddles connector with winding number $N$. 

Figure~\ref{fig:green} illustrates the case $N=0$. In this specific case, \(\mathcal{W}_{E'}^-\) is an orbit with winding number \( w_{E'} =0\). It connects to \(N_{E'}^+\) in our fundamental domain, and since \(\mathcal{W}_{E'}^-\) lies above \(\sigma_{E'}\), \(A(E')<0\). On the other hand, \(\mathcal{W}_{E''}^-\), an orbit of winding number \( w_{E''} \geq 1\) lies beneath \(\sigma_{E''}\), so \(A(E'')>0\). 

It remains to show that \(A(E)\) is a continuous function of \(E\).  
To show that, let \(E_n \in [E',E'']\) be any sequence such that \(E_n \rightarrow E\). Since our trajectories depend continuously on the parameter \(E\), we have that \(y_{E_n}^{\pm} \rightarrow y_{E}^{\pm}\) pointwise. Since \(y_{E}^{\pm}\) is monotone in $E$, both \(y_{E_n}^{-}\) and \(y_{E_n}^{+}\) are bounded uniformly. Therefore, \(A(E_n) \rightarrow A(E)\) by Lebesgue's dominated convergence theorem.
\end{proof}

\begin{figure}[H]
    \centering
    \includegraphics[scale=0.28]{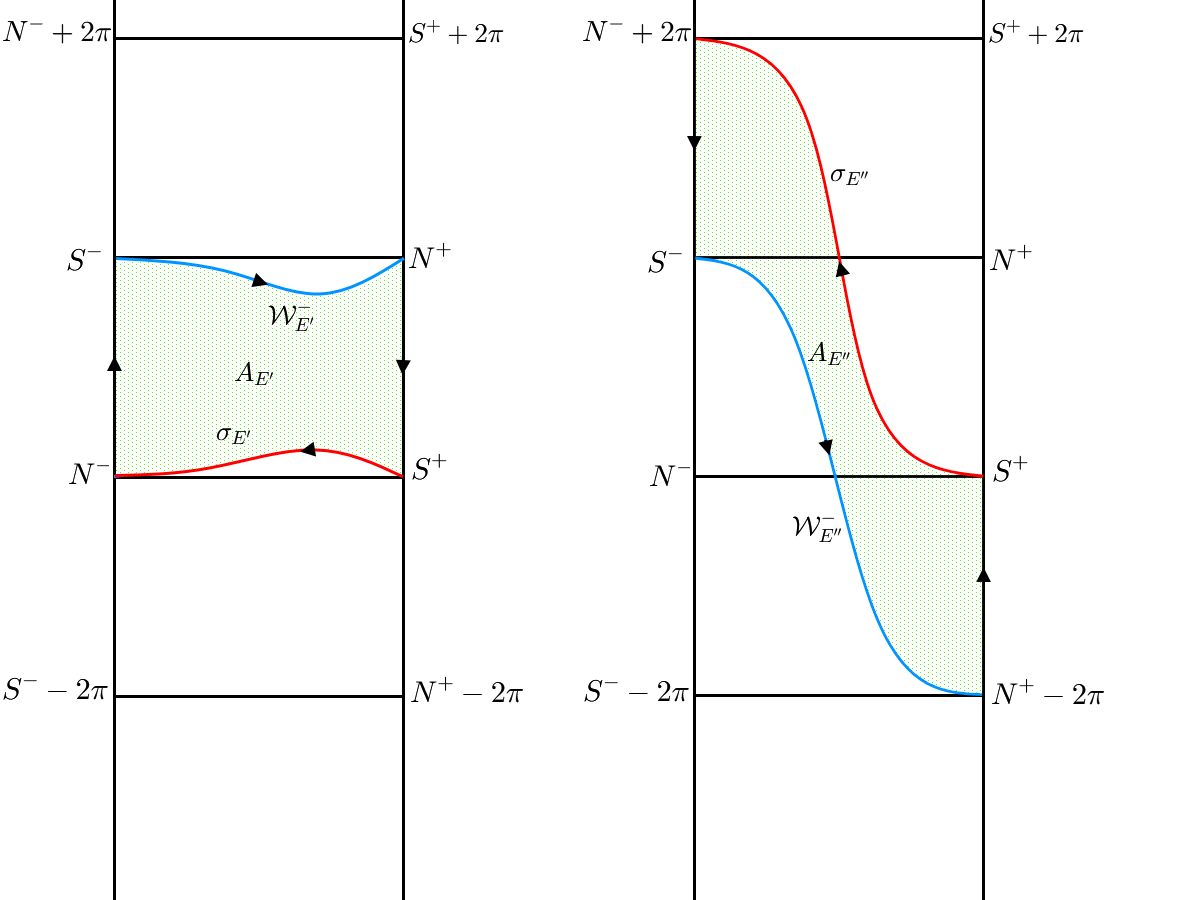}
    \caption{Area sign change}
    \label{fig:green}
\end{figure}

\subsection{Construction of Barriers}\label{barriers}
 Using barriers, we will prove the existence of an orbit with winding number $\leq N$ and prove the existence of another orbit with winding number $\geq N+1$. Then we will apply Theorem \ref{Theorem: GT} to prove the existence of a saddles connector with winding number $N$. 

First we show a general result about orbits of more energy being a lower barrier for orbits of less energy:
\begin{proposition}\label{prop:barriers}
    Let $-1\leq E_1<E<E_2\leq 1$.  Let $\mathcal{W}_E^-$ denote the unique orbit of the system \eqref{zThetasys} whose $\alpha$-limit is $S_E^- = (-\frac{\pi}{2},\cos^{-1}(E))$.  Then $\mathcal{W}_{E_1}^-$ is an upper barrier (as defined below) for $\mathcal{W}_E^-$ and similarly $\mathcal{W}_{E_2}^-$ is a lower barrier for $\mathcal{W}_{E}^-$.
\end{proposition}
\begin{proof}
    Let $(z(\tau),\Theta_i(\tau))$ be the two orbits $\mathcal{W}_{E_i}^-$, for $i=1,2$. To prove the statement, we need to compare the slope of the orbit $\mathcal{W}_E^-$ with the slopes of these. We have

\begin{eqnarray*}
\left.\frac{d \Theta}{dz}\right|_{\Theta = \Theta_1} - \frac{d\Theta_1}{dz} & = & \sec^2(z) \left( \left.\frac{d \Theta}{d\tau}\right|_{\Theta = \Theta_1} - \frac{d\Theta_1}{d\tau} \right) \\
& = & \sec^2(z) (G_E(z,\Theta_1) - G_{E_1}(z,\Theta_1)) = -2 (E -E_1) < 0.
\end{eqnarray*}
Thus, if the orbit $\mathcal{W}_E^-$ were to cross $\mathcal{W}_{E_1}^-$, it could only cross it from above to below.
Moreover, since $\cos^{-1}$ is a decreasing function, the $\alpha$-limit of $\mathcal{W}_E^-$ is clearly below that of $\mathcal{W}_{E_1}^-$, therefore it is impossible for $\mathcal{W}_E^-$ to ever end up above $\mathcal{W}_{E_1}^-$.  In this sense $\mathcal{W}_{E_1}^-$ is an upper barrier for $\mathcal{W}_E^-$.  This proof can also be used to show that $\mathcal{W}_{E_2}^-$ is a lower barrier.
\end{proof}
By the above proposition, if we can prove the existence of an orbit with $E=1$ that connects an equilibrium point on the left-hand side of the cylinder with another on the right-hand side, since that is the highest value of energy possible, that orbit would acts as ``the mother of all floors," meaning it could be used as a universal lower barrier for all saddles connectors.  This is accomplished in the next theorem.

\begin{theorem}\label{ThmWN}
For \(E=1\), there exists a sequence of values $0= \gamma_0<\gamma_1<\gamma_2<\gamma_3<\dots$, so that if $\gamma >0$ and $\gamma\in [\gamma_{k-1}, \gamma_k)$ for some integer $k\geq 1$, then there exists a heteroclinic orbit for system \eqref{zThetasys} with winding number $k$ and another with winding number $k+1$.
\end{theorem}
\begin{figure}[H]
    \centering
    \includegraphics[scale=0.372]{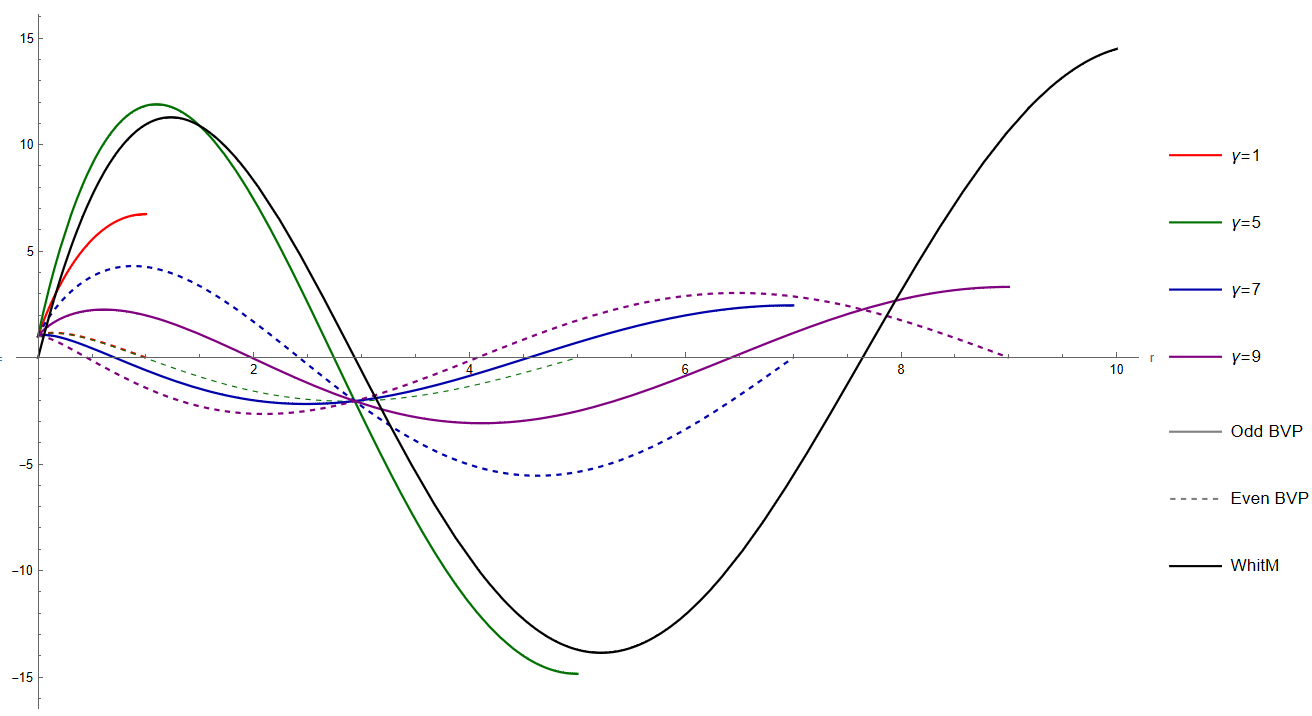}
    \caption{A plot of a scaled version of $-iM_{-i,\frac{1}{2}}(ir)$ vs $r$ along with plots of solutions to the odd and even boundary value problems for several values of $\gamma$.  \label{fig:whit2}}
    
\end{figure}
\begin{proof}
We set $E=1$ and rewrite the system \eqref{eq:uv}:
\begin{equation}
    \left\{\begin{array}{rcl}
        \frac{du}{ds} & = &\left(2+\frac{\gamma}{2} e^{-|s|}\right)v,  \\
         \frac{dv}{ds} & = & - \frac{\gamma}{2} e^{-|s|} u.
    \end{array}\right.
\end{equation}
From the second equation, $u = -\frac{2}{\gamma} e^{|s|} \frac{dv}{ds}$.  Plugging that into the first equation, we obtain a second order linear ODE for $v(s)$:
\begin{equation}
    \frac{d^2v}{ds^2} + \frac{s}{|s|} \frac{dv}{ds} + \left( \gamma e^{-|s|} + \frac{\gamma^2}{4} e^{-2|s|} \right) v = 0.
\end{equation}
We observe that changing $s$ to $-s$ leaves this equation invariant.  It is therefore enough to solve the above on $(-\infty,0)$ and then extend the function $v$ to all of $\mathbb{R}$ as an even function. Thus the equation to solve is
\begin{equation}\label{eq:v}
    v'' - v' + \left( \gamma e^{s} + \frac{\gamma^2}{4} e^{2s} \right) v = 0,\qquad -\infty<s<0,
\end{equation}
where prime denotes differentiation with respect to $s$.
Having found $v$ for $s<0$, one can then solve for $u$ by setting
\begin{equation}\label{eq:u}
    u(s) = - \frac{2}{\gamma} e^{-s} \frac{dv}{ds},\qquad -\infty<s<0.
\end{equation}
Since the extended $v$ is even, the extended $u$ has to be an odd function, so we extend $u$ to all of $\mathbb{R}$ as an odd function.  Note that $v$ may not be differentiable at $s=0$, and thus $u$ may have a jump discontinuity there.

Once $u$ and $v$ are found in this way, one can compute $\Theta = 2\tan^{-1}\left(\frac{v}{u}\right)$ and verify that it has the requisite winding number.

To solve \eqref{eq:v}, we make a change of variable that transforms it into a known equation:  Let $r = \gamma e^s$ and define $w(r) = v(s)$. We then have $v_s = r w_r$ and $v_{ss} = r^2 w_{rr} + r w_r$.  We therefore obtain from \eqref{eq:v} that
\begin{equation}\label{eq:vWhit}
    \ddot{w} + \left( \frac{1}{4} + \frac{1}{r} \right) w = 0
\end{equation}
which is known as Whittaker's equation (see e.g. \cite{NIST:DLMF}, \S 13.14,) with parameters $\kappa = -i$ and $\mu = \frac{1}{2}$. (To see that, change the independent variable to $x=ir$.) Here, $\dot{w}$ is differentiation with respect to $r$.

The general solution of Whittaker's equation is a (complex) linear combination of the two Whittaker functions\footnote{Whittaker functions of imaginary argument are also known as {\em Coulomb wave functions} (See \S \ref{CWFs} for a definition.)  All of the subsequent analysis here can be equivalently formulated in terms of Coulomb wave functions, which may have the advantage of being real-valued, but we choose to work with Whittaker functions since they are more well-known.} $M_{\kappa,\mu}$ and $W_{\kappa,\mu}$.  We thus have that the general solution to \eqref{eq:vWhit} is
\begin{equation}\label{eq:vgensol}
    w_{gen}(r) = c_1 M_{-i,\frac{1}{2}}(ir) + c_2 W_{-i,\frac{1}{2}}(ir),\qquad c_1,c_2\in \mathbb{C}.
\end{equation}
To find $c_1$ and $c_2$ we need to supplement \eqref{eq:vWhit} with two boundary conditions.  These need to be set in such a way that the corresponding solution for the $\Theta$ equation has a desired winding number. We accomplishing this by making sure $v$ and $u$ have asymptotic behaviors as $s \to \pm\infty$ that are compatible with the heteroclinic orbit beginning and ending at the right equilibrium points.

Recall that the equilibrium point on the left side of the cylinder corresponds to $s=-\infty$, and therefore to $r=0$.  We use the known asymptotic behavior at zero of the Whittaker functions that show up in \eqref{eq:vgensol}:
\begin{equation}
    M_{-i,\frac{1}{2}}(z) = z(1+O(z))\mbox{ as }z\to 0,\qquad W_{-i,\frac{1}{2}}(z) = \frac{1}{\Gamma(1+i)} + O(z \ln z)\mbox{ as }z\to 0
\end{equation}
($\Gamma$ is the Gamma function.)

It thus follows that the general solution \eqref{eq:vgensol} goes to a constant value $v_0 := c_2/\Gamma(1+i)$ as $r\to 0$, which would be nonzero if $c_2 \ne 0$, so that $v(s) \sim v_0\ne 0$ as $s\to -\infty$.  From the equation satisfied by $u(s)$, namely
\begin{equation}\label{eq:uprime}
    \frac{du}{ds} = (2+\frac{\gamma}{2} e^{s}) v
\end{equation}
it follows that as $s \searrow -\infty$, we have $\frac{du}{ds} \sim 2 v_0$, and thus $u(s) \sim 2v_0 s$, so that $\Theta(s) \sim 2 \tan^{-1}\frac{1}{2s}$ and thus by choosing a branch of arctan we can arrange it that $\Theta \nearrow 2\pi$ as $s \searrow -\infty$.  Thus the corresponding $\Theta$ orbit has the correct $\alpha$-limit, and the only condition we find on $v$ is that $v_0\ne 0$, which can be assured by choosing the boundary condition $v(0) = 1$ for \eqref{eq:vWhit}. We also note that for $E=1$ the function $G_E(z,\Theta)$ is negative everywhere, so that $\Theta$ would be a monotone decreasing function of $z$.

Thus, since the only equilibrium points of the system \eqref{zThetasys} are at $(\pm\frac{\pi}{2},2\pi\mathbb{Z})$, the $\omega$ limit of this orbit is $\Theta(\infty) = -2\pi n$ for some integer $n\geq 0$.  The winding number of the orbit is thus $N = n+1$.

Recall that $v$ is an even function of $s$, and $u$ is an odd function of $s$.  It follows that the $\Theta$ orbit must be symmetric with respect to $s=0$ and therefore $\Theta(0) = \frac{1}{2}(\Theta(-\infty)+\Theta(\infty)) = (2-N)\pi$.

Now suppose $N$ is odd. Then $\tan\frac{\Theta(0)}{2} = \pm \infty$, which can be achieved if $u(0) = 0$, i.e. $\frac{dv}{ds}(s=0) = 0$. Thus, to have an orbit with an odd winding number, we may choose the other boundary condition for \eqref{eq:vWhit} on the interval $[0,\gamma]$ to be $\dot{w}(\gamma) = 0$  (recall that $r=\gamma e^s$ so $s=0$ corresponds to $r=\gamma$.) We thus have the following boundary value problem for \eqref{eq:vWhit}:
\begin{equation}\label{eq:bvpodd}
    \ddot{w} + \left( \frac{1}{4} + \frac{1}{r} \right) w = 0,\qquad w(0) = 1,\quad \dot{w}(\gamma) = 0.
\end{equation}

Suppose on the other hand that $N$ is even.  We then have $\tan\frac{\Theta(0)}{2} = 0$, which can be achieved if $v(s=0) = 0$.  We thus obtain another boundary value problem for \eqref{eq:vgensol} on the interval $[0,\gamma]$ in that case:
\begin{equation}\label{eq:bvpeven}
    \ddot{w} + \left( \frac{1}{4} + \frac{1}{r} \right) w = 0,\qquad w(0) = 1,\quad {w}(\gamma) = 0.
\end{equation}
Both of the above boundary value problems we can solve since we know the general solution \eqref{eq:vgensol}.  For the $N$ odd case, we obtain:
\begin{equation}\label{eq:oddsol}
    w(r) = -\Gamma(1+i)\frac{W'_{-i,1/2}(i\gamma)}{M'_{-i,1/2}(i\gamma)} M_{-i,1/2}(ir) + \Gamma(1+i) W_{-i,1/2}(ir),
\end{equation}
 with prime denoting differentiation with respect to the argument of the Whittaker functions. This is a valid solution on $[0,\gamma]$ provided the denominator $M'_{-i,1/2}(i\gamma)$ does not vanish.  Similarly, for the case $N$ even we find
\begin{equation}\label{eq:evensol}
    w(r) = -\Gamma(1+i)\frac{W_{-i,1/2}(i\gamma)}{M_{-i,1/2}(i\gamma)} M_{-i,1/2}(ir) + \Gamma(1+i) W_{-i,1/2}(ir),
\end{equation}
which is once again valid on $[0,\gamma]$ provided $M_{-i,1/2}(i\gamma) \ne 0$.  

Note that despite the appearance of complex numbers in these solutions, they must be real, since they are solutions to real boundary value problems for linear equations with real coefficients.  Complex coefficients appear because Whittaker functions themselves are complex-valued. 

Let us therefore define the following increasing sequence $\gamma_k$ of real numbers, with $\gamma_0=0$ and 
\begin{equation}\label{def:gammaj}
     \left\{\begin{array}{rcl}
        \gamma_{2j-1} &=& \mbox{$j$-th positive root of $M'_{-i,1/2} $} \\
         \gamma_{2j} & = & \mbox{$j$-th positive root of $M_{-i,1/2}$}
    \end{array}\right. \quad j=1,2,3,\dots
\end{equation}
For example, we can numerically compute the first few of these to be 
\begin{equation}
    \gamma_1 = 1.230870178, \gamma_2 = 2.934791015, \gamma_3 = 5.218667468, \gamma_4 = 7.643742568.
\end{equation}

For $\gamma \in [\gamma_{k-1}, \gamma_{k})$, we have $\gamma < \gamma_k$ and $\gamma<\gamma_{k+1}$ so the boundary value problems in \eqref{eq:bvpodd} and \eqref{eq:bvpeven} both have valid solutions on $r \in [0,\gamma]$. We need to figure out what the corresponding winding numbers of these would  be. We first consider the boundary value problem \eqref{eq:bvpodd} when $N$ is odd. 

Now, recall that

\begin{equation}
\begin{cases}
        w(r) =  w(\gamma e^s) = v(s) & \text{if } s <0,\\
         v(s) = v(-s)  & \text{if } s > 0,
\end{cases}
\end{equation}
since $r=\gamma e^s$ so $s=0$ corresponds to $r=\gamma$.






For the odd boundary value problem, we have that
\begin{equation}
    \begin{cases}
    v'(0) = 0 \\
    v(s \rightarrow -\infty) = 1 \\
    v(s \rightarrow \infty) = 1.
    \end{cases}
\end{equation}
In addition to this, $v'(s)$ will be zero $2\lfloor\frac{k}{2}+1\rfloor-1$ times. This is because for $r\leq\gamma$, the general solution $w(r)$ will have exactly $\lfloor\frac{k}{2}+1\rfloor$ critical points by Lemma \ref{lem:Lem1} (see Appendix). $w(r)$ on $r\in(0,\gamma]$ and $v(s)$ on $s\in(-\infty, 0)$ must share the same number of critical points since $v'(s)=\gamma e^s w'(\gamma e^s) = r w'(r)$. Since $v(s)$ is even, we double this number and subtract the one at $r=0$ to get the number of critical points of $v(s)$ for $s\in(-\infty, \infty)$. Then, since
\begin{equation}
    u(s) = - \frac{2}{\gamma} e^{-s} \frac{dv}{ds},\qquad -\infty<s<0,
\end{equation}
and we know that $u(s) = -u(-s)$ and $u(s) \sim 2v_{0}s$, this means that 
\begin{equation}
    \begin{cases}
    u(s\rightarrow -\infty) = -\infty  \\
    u(0) = 0 \\
    u(s\rightarrow \infty) =  \infty.
    \end{cases}
\end{equation}
So, $u(s)$ will also be zero $2\lfloor\frac{k}{2}+1\rfloor-1$ times since $u\propto \frac{dv}{ds}$ by \eqref{eq:u}. 

From here, since $\Theta = 2\tan^{-1}{\frac{v}{u}}$ we determine that $\Theta (s\rightarrow \infty) = -2\pi (2\lfloor\frac{k}{2}+1\rfloor-1)$. That is, $\Theta$ passes through the required number of branches of $\tan^{-1}$ to match the number of times $\frac{v}{u}$ diverges for $s\in (-\infty,\infty)$. Therefore, 
\begin{equation}
    \begin{cases}
    \Theta (s\rightarrow -\infty) = 2\pi  \\
    \Theta (s\rightarrow \infty) =  -2\pi (2\lfloor\frac{k}{2}+1\rfloor-1).
    \end{cases}
\end{equation}
By definition, the winding number of such an orbit is $ N = \frac{\Theta(-\infty)-\Theta(\infty)}{2\pi} = 2\lfloor\frac{k}{2}+1\rfloor-1$. 

Now we consider the boundary value problem \eqref{eq:bvpeven} for an orbit with $N$ even. 




Here, since $v(0)=w(\gamma) = 0$, we see that:
\begin{equation}
    \begin{cases}
    v(0) = 0 \\
    v(s \rightarrow -\infty) = 1 \\
    v(s \rightarrow \infty) = 1.
    \end{cases}
\end{equation}

In addition to this, $v'(s)$ will be zero $2\lfloor \frac{k}{2}+\frac{1}{2} \rfloor$ times. This is because, for $r\leq\gamma$, the general solution $w(r)$ will have exactly $\lfloor \frac{k}{2}+\frac{1}{2} \rfloor$ critical points by Lemma \ref{lem:Lem1}. Like  before, $w(r)$ on $r\in(0,\gamma]$ and $v(s)$ on $s\in(-\infty, 0)$ must share the same number of critical points. Since $v(s)$ is even, we double this number to get the number of critical points of $v(s)$ for $s\in(-\infty, \infty)$.

Likewise, $u(s)$ will also be zero $2\lfloor 
\frac{k}{2}+\frac{1}{2} \rfloor$ times and so,
\begin{equation}
    \begin{cases}
    \Theta (s\rightarrow -\infty) = 2\pi  \\
    \Theta (s\rightarrow \infty) =  -2\pi (2\lfloor 
\frac{k}{2}+\frac{1}{2} \rfloor).
    \end{cases}
\end{equation}
By definition, the winding number of such an orbit is $ N = \frac{\Theta(-\infty)-\Theta(\infty)}{2\pi} = 2\lfloor 
\frac{k}{2}+\frac{1}{2} \rfloor$. 

Thus, there exists a heteroclinic orbit for the system \eqref{zThetasys} with winding number $2\lfloor\frac{k}{2}+1\rfloor-1$, and another one with winding number $2\lfloor\frac{k}{2}+\frac{1}{2}\rfloor$. More simply, this means there exists an orbit with winding number $k$ and another with winding number $k+1$. Once these exist, there cannot be a third orbit with a higher winding number, as it would necessarily intersect with at least one of these two orbits, which would violate the existence and uniqueness theorem for solutions of ODEs. 
\end{proof}

An analogous result holds for $E=-1$ as well:

\begin{theorem}\label{ThmWN2}
For \(E= -1\), there exists a sequence of values $0= \Gamma_0<\Gamma_1<\Gamma_2<\Gamma_3<\dots$, so that if $\gamma>0$ and $\gamma\in [\Gamma_{j-1}, \Gamma_j)$, for some $j \geq 1$, then there are exactly two heteroclinic orbits of the system \eqref{zThetasys}, one with winding number $j-2$ and another one with winding number $j-1$.
\end{theorem}
\begin{proof}
We set $E= -1$ and rewrite the system \eqref{eq:uv}:
\begin{equation}
    \left\{\begin{array}{rcl}
        \frac{du}{ds} & = &\left(\frac{\gamma}{2} e^{-|s|}\right)v  \\
         \frac{dv}{ds} & = & (2- \frac{\gamma}{2} e^{-|s|}) u.
    \end{array}\right.
\end{equation}

From the first equation, $v = \frac{2}{\gamma} e^{|s|} \frac{du}{ds}$.  Plugging that into the second equation, we obtain a second order linear ODE for $u(s)$:
\begin{equation}
    \frac{d^2u}{ds^2} + \frac{s}{|s|} \frac{du}{ds} + \left( -\gamma e^{-|s|} + \frac{\gamma^2}{4} e^{-2|s|} \right) u = 0.
\end{equation}
We observe that changing $s$ to $-s$ leaves this equation invariant.  It is therefore enough to solve the above on $(-\infty,0)$ and then extend the function $u$ to all of $\mathbb{R}$ as an even function. Thus the equation to solve is
\begin{equation}\label{eq:u2}
    u'' - u' + \left( -\gamma e^{s} + \frac{\gamma^2}{4} e^{2s} \right) u = 0,\qquad -\infty<s<0,
\end{equation}
where prime denotes differentiation with respect to s.
Having found $u$ for $s<0$, one can then solve for $v$ by setting
\begin{equation}
    v(s) = \frac{2}{\gamma} e^{-s} \frac{du}{ds},\qquad -\infty<s<0.
\end{equation}
Since the extended $u$ is even, the extended $v$ has to be an odd function, so we extend $v$ to all of $\mathbb{R}$ as an odd function.  
Once $u$ and $v$ are found in this way, one can compute $\theta = 2\tan^{-1}\left(\frac{v}{u}\right)$ and verify that it has the requisite winding number.

To solve \eqref{eq:u2}, we make the same change of variable that transforms it into a known equation:  Let $r = \gamma e^s$ and $w(r) = u(s)$.  We then obtain from \eqref{eq:u2} that
\begin{equation}\label{eq:newWhit}
    \ddot{w} +\left(  \frac{1}{4}-\frac{1}{r}\right) w = 0,
\end{equation}
where $\ddot{w}$ is the second derivative of $w$ with respect to $r$. This is Whittaker's equation, with parameters $\kappa = i$ and $\mu = \frac{1}{2}$, which can be seen with the change of variables $x=ir$.

The general solution of Whittaker's equation is a (complex) linear combination of the two Whittaker functions $M_{\kappa,\mu}$ and $W_{\kappa,\mu}$.  We thus have that the general solution to \eqref{eq:newWhit} is
\begin{equation}\label{eq:ugensol}
    w_{gen}(r) = c_1 M_{i,\frac{1}{2}}(ir) + c_2 W_{i,\frac{1}{2}}(ir),\qquad c_1,c_2\in \mathbb{C}.
\end{equation}
To find $c_1$ and $c_2$ we need to supplement \eqref{eq:vWhit} with two boundary conditions.  These need to be set in such a way that the corresponding solution for the $\Theta$ equation has a desired winding number. We accomplish this by making sure $v$ and $u$ have asymptotic behaviors as $s \to \pm\infty$ that are compatible with the heteroclinic orbit beginning and ending at the right equilibrium points.

Recall that the equilibrium point on the left side of the cylinder corresponds to $s=-\infty$, and therefore to $r=0$.  We use the known asymptotic behavior at zero of the Whittaker functions that show up in \eqref{eq:ugensol}:
\begin{equation}
    M_{i,\frac{1}{2}}(z) = z(1+O(z))\mbox{ as }z\to 0,\qquad W_{i,\frac{1}{2}}(z) = \frac{1}{\Gamma(1-i)} + O(z \ln z)\mbox{ as }z\to 0
\end{equation}
($\Gamma$ is the Gamma function.)

It thus follows that the general solution \eqref{eq:ugensol} goes to a constant value $u_0 := c_2/\Gamma(1-i)$ as $r\to 0$, which would be nonzero if $c_2 \ne 0$, so that $u(s) \sim u_0\ne 0$ as $s\to -\infty$.  From the equation satisfied by $v(s)$, namely
\begin{equation}
    \frac{dv}{ds} = (2-\frac{\gamma}{2} e^{s}) u,
\end{equation}
it follows that as $s \searrow -\infty$, we have $\frac{dv}{ds} \sim 2 u_0$, and thus $v(s) \sim 2u_0 s$, so that $\Theta(s) \sim 2 \tan^{-1}{2s}$ and thus by choosing a branch of arctan we can arrange it that $\Theta \searrow 3\pi$ as $s \searrow -\infty$.  Thus the corresponding $\Theta$ orbit has the correct $\alpha$-limit, and the only condition we find on $u$ is that $u_0\ne 0$, which can be assured by choosing the boundary condition $w(0) = 1$ for \eqref{eq:newWhit}. Also, we note that, in contrast to the $E=1$ case, here $\Theta$ will be an {\em increasing} function of $z$ in a neighborhood of the endpoints.

Since the only equilibrium points of the system \eqref{zThetasys} are at $(\pm\frac{\pi}{2}, \pi+2\pi\mathbb{Z})$, the $\omega$ limit of this orbit is $\Theta(\infty) = 3\pi-2\pi n$ for some integer $n\in \mathbb{Z}$.  The winding number of the orbit is thus $N=n$ (which can be negative).

 Recall that $u$ is an even function of $s$, and $v$ is an odd function of $s$.  It follows that the $\Theta$ orbit must be symmetric with respect to $s=0$ and therefore $\Theta(0) = \frac{1}{2}(\Theta(-\infty)+\Theta(\infty)) = 3\pi - \pi N $.

Now suppose $N$ is odd. Then $\tan\frac{\Theta(0)}{2} = 0$, which can be achieved if $v(0) = 0$, i.e. $\frac{du}{ds}(s=0) = 0$. Thus, to have an orbit with an odd winding number, we may choose the other boundary condition for \eqref{eq:newWhit} on the interval $[0,\gamma]$ to be $\dot{w}(\gamma) = 0$  (recall that $r=\gamma e^s$ so $s=0$ corresponds to $r=\gamma$.) We thus have the following boundary value problem for \eqref{eq:newWhit}:
\begin{equation}\label{eq:bvpodd2}
     \ddot{w} +\left(  \frac{1}{4}-\frac{1}{r}\right) w = 0,\qquad w(0) = 1,\quad \dot{w}(\gamma) = 0.
\end{equation}

Suppose on the other hand that $N$ is even.  We then have $\tan\frac{\Theta(0)}{2} = \pm \infty$, which can be achieved if $u(s=0) = 0$.  We thus obtain another boundary value problem for \eqref{eq:vgensol} on the interval $[0,\gamma]$ in that case:
\begin{equation}\label{eq:bvpeven2}
    \ddot{w} +\left(  \frac{1}{4}-\frac{1}{r}\right) w = 0,\qquad w(0) = 1,\quad w(\gamma) = 0.
\end{equation}
Both of the above boundary value problems we can solve since we know the general solution \eqref{eq:ugensol}.  For the $N$ odd case, we obtain (with prime denoting differentiation with respect to the argument of the Whittaker functions):
\begin{equation}
    w(r) = -\Gamma(1-i)\frac{W'_{i,1/2}(i\gamma)}{M'_{i,1/2}(i\gamma)} M_{i,1/2}(ir) + \Gamma(1-i) W_{i,1/2}(ir),
    \label{eq:oddsol2}
\end{equation}
which is a valid solution on $[0,\gamma]$ provided the denominator $M'_{-i,1/2}(i\gamma)$ does not vanish.  Similarly, for the case $N$ even we find
\begin{equation}
    w(r) = -\Gamma(1-i)\frac{W_{i,1/2}(i\gamma)}{M_{i,1/2}(i\gamma)} M_{i,1/2}(ir) + \Gamma(1-i) W_{i,1/2}(ir),
     \label{eq:evensol2}
\end{equation}
which is once again valid on $[0,\gamma]$ provided $M_{i,1/2}(i\gamma) \ne 0$.  
Note that despite the appearance of complex numbers in these solutions, they must be real, since they are solutions to real boundary value problems for linear equations with real coefficients.  Complex coefficients appear because Whittaker functions themselves are complex-valued. 

Let us therefore define the following increasing sequence $\Gamma_k$ of real numbers, with $\Gamma_0=0$ and 
\begin{equation}\label{def:Gammaj}
     \left\{\begin{array}{rcl}
        \Gamma_{2j-1} &=& \mbox{$j$-th positive root of $M'_{i,1/2} $} \\
         \Gamma_{2j} & = & \mbox{$j$-th positive root of $M_{i,1/2}$}
    \end{array}\right. \quad j=1,2,3,\dots
\end{equation}
For example, we can numerically compute the first few of these to be 
\begin{equation}
    \Gamma_1 = 7.3148, \Gamma_2 = 11.6282, \Gamma_3 = 15.3354, \Gamma_4 = 18.9491.
\end{equation}
Let $\gamma>0$ be given.  There exists an integer $j\geq 1$ such that $\gamma \in [\Gamma_{j-1},\Gamma_j)$.  We therefore have that $\gamma< \Gamma_j$ and $\gamma < \Gamma_{j+1}$, so that both boundary value problems in the above have valid solutions.  

  We prove in Lemma \ref{lem:Lem1} (see Appendix) that for $E=-1$,  $w(r)$ will have $\lfloor \frac{j}{2}  \rfloor$ zeroes between $r\in [0,\gamma]$ in the solution to \eqref{eq:bvpodd2} and $\lfloor \frac{j}{2} + \frac{1}{2}\rfloor$ zeroes between $r\in [0,\gamma]$ in the solution to \eqref{eq:bvpeven2}. Consequently, the number of zeros of $u(s)$ would be 
  $2 \lfloor \frac{j}{2} \rfloor$ and $2\lfloor \frac{j}{2}+\frac{1}{2} \rfloor-1$ respectively, noting that the $1$ is subtracted in the latter case to avoid double counting the zeroes of $u(s)$ at $s=0$ due to our boundary condition. 
  
  Consider the case $j=1$.  In that case, $u$ will have no zeros, which implies that the branch of arctan to be chosen goes from $3\pi/2$ to $5\pi/2$, or in other words $\Theta$ goes from $3\pi$ up to $5\pi$, so that the winding number of the corresponding orbit is $N=-1$. From here we deduce that in general, the winding number of corresponding $E=-1$ orbits will be $2 \lfloor \frac{j}{2} \rfloor -1$ and $2\lfloor \frac{j}{2}+\frac{1}{2} \rfloor-2$, respectively.
  
  It follows that for $\gamma \in [\Gamma_{j-1},\Gamma_j)$ and $E = -1$, there exists a heteroclinic orbit for the system \eqref{zThetasys} with winding number  $j-2$ and another with winding number $j-1$. As in the case of $E=1$, there cannot be a third orbit with a higher winding number, as it would necessarily intersect with at least one of these two orbits which would violate the existence and uniqueness theorem for solutions of ODEs. 

Thus, for any $\gamma>0$ there exists exactly two orbits of \eqref{zThetasys} with $E= -1$. 
\end{proof}

We can now state and prove the main result of this paper:
\begin{theorem}\label{Theorem:main}
Let $\gamma>0$ be given, and let $\{\gamma_j\}, \{\Gamma_j\}$ be defined by \eqref{def:gammaj}, \eqref{def:Gammaj} respectively.  Thus there exist integers $n\geq 0$ and $j\geq n+1$ such that $\gamma \in [\gamma_{j-1},\gamma_j) \cap[\Gamma_{n}, \Gamma_{n+1})$.  Let $k(n)\geq 1$ be the smallest integer with the property that $\gamma_{n+k(n)} > \Gamma_{n}$.
Then 
\begin{enumerate}
\item The dynamical system
\eqref{zThetasys} has exactly $N := j - n$ saddles connectors with all winding numbers between $n$ and $j-1$. 
\item The hamiltonian \eqref{eq:Ham} has a ground state with winding number $n$ and a maximum of $k(n+1)-1$ excited states with higher winding numbers. 
\item The discrete spectrum of the hamiltonian consists of finitely many simple eigenvalues $E_i$ with $$-1<E_0< E_1 < \dots < E_{N-1}<1.$$
\end{enumerate}
\end{theorem}
\begin{figure}[H]
    \centering
    \includegraphics[width=0.5\textwidth]{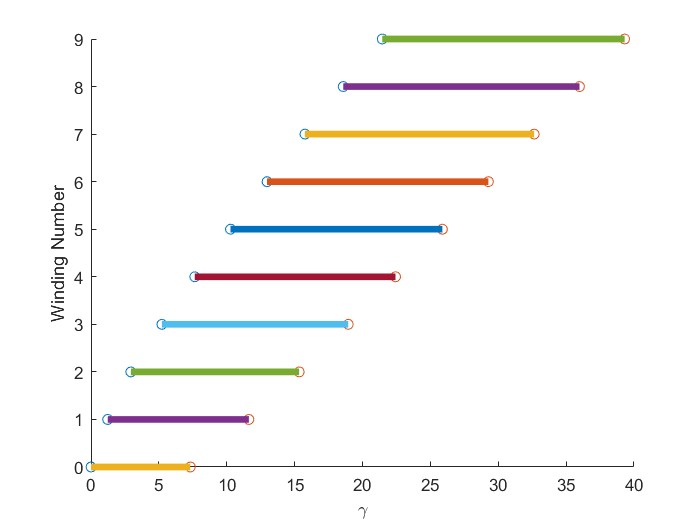}\hfill
    \includegraphics[width=0.5\textwidth]{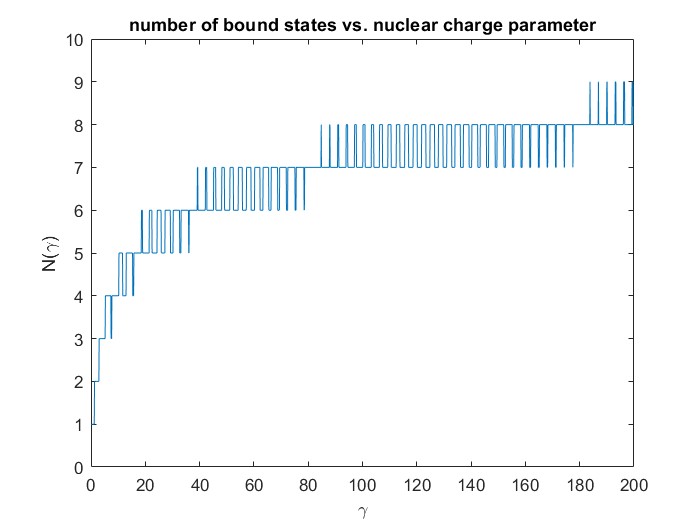}
    \caption{Winding numbers of saddles connectors and number of bound states versus $\gamma$}
    \label{fig:staircase}
\end{figure}
\begin{proof}

For $j\geq 1$ let $I_j := [\gamma_{j-1},\gamma_j)$ and $J_j := [\Gamma_{j-1},\Gamma_j)$. These are two families of disjoint intervals covering $[0,\infty)$.  
Let $\gamma>0$ be given.  Then there is a unique $n\geq 0$ such that $\gamma \in J_{n+1}$ and there is a unique integer $j \geq 1$ such that $\gamma \in I_j$.  In Lemma 1 of the Appendix we show that $\gamma_i < \Gamma_{i}$ for all $i\geq 1$, which implies  that  we must have $$n+1\leq j \leq n+k(n+1).$$ 
Therefore Theorem \ref{ThmWN} guarantees the existence of orbits with $E=1$ and winding numbers $j$ and $j+1$, while Theorem \ref{ThmWN2} guarantees the existence of orbits with $E=-1$ and winding numbers $n-1$ and $n$. We also know that no other orbits of different winding numbers can exist at these values of energy.
Since there exists an orbit with winding number $n$ for $E=-1$ and an orbit with winding number $j\geq n+1$ for $E=1$, 
By Theorem 1, this means that there exists some $E \in (-1,1)$ such that there exists a saddles connector $W_n$ with winding number $n$. Since the $E=1$ orbits act as universal lower barriers, the existence of the $E=1$ orbit with winding number $j$ implies that there can be no saddles connector with winding number $j$ or higher. Similarly, since $E=-1$ orbits function as universal upper barriers, the existence of an $E=-1$ orbit with winding number $n$ rules out the possibility of a saddles connector with a winding number $n-1$ or lower as well. Otherwise, all winding numbers between $n$ and $j-1$ are allowed for saddle connectors, and their existence can be established by repeated use of Theorem 1.

It follows that for $\gamma \in J_n \cap I_j$, the hamiltonian \eqref{eq:Ham} has a total of $N:= j-n$ bound states, with the ground state having winding number $n$. Since $j \leq n+k(n+1)$, the maximum number of bound states is $k(n+1)$.  See Fig.~\ref{fig:staircase}.  

Finally, let $E_i\in(-1,1)$ denote the energy eigenvalue of the eigenfunction that corresponds to the saddles connector with winding number $i$.  If $E_{i+1}<E_i$, it would follow from Prop.~\ref{prop:barriers} that $W_{i}$ would always lie below $W_{i+1}$, which is a contradiction since the $\omega$-limit of $W_i$ is $-2\pi i -\cos^{-1}E_i$ and the $\omega$-limit of $W_{i+1}$ is $-2\pi i - 2\pi -\cos^{-1}E_{i+1}$ so $W_{i+1}$ needs to cross $W_{i}$ and go below it. Therefore $E_i \leq E_{i+1}$.  We already know by Prop.~\ref{prop:Widemann} that all eigenvalues have to be simple, so $E_i \ne E_{i+1}$. This establishes the claim.
\end{proof}

\section{Numerical Investigations}\label{Numerical Section}
\subsection{Computing the discrete energy levels}
In this section we use the computational software package Matlab \cite{MATLAB} to create a binary search program that will allow us to input an initial guess $E_N^0$ for the actual energy eigenvalue $E_N$ of a saddles connector with a specified winding number $N$, and a tolerance $\epsilon$, and it will search for an approximate eigenvalue $E^{\mbox{\tiny app}}_N$ such that $|E^{\mbox{\tiny app}}_N - E_N| < \epsilon$.

In order to find these approximate values, we first numerically solve the $\Theta$ equation
\begin{equation}
    \Theta' = 2\cos{\Theta} - \gamma e^{-|s|} - 2E,
    \label{numericalode}
\end{equation}
for the given $E=E_N^0$ and the initial value 
$$ \Theta(0) = \Theta_N= \frac{1}{2} (\Theta(-\infty) + \Theta(\infty)) = \frac{1}{2}( 2\pi + \cos^{-1}E + 2\pi - \cos^{-1}E - 2\pi N) = \pi(2-N).$$
(Here we are using the fact that the solution to the above equation will be symmetric under reflection with respect to $(s=0,\Theta = \Theta(0))$.)

The parameter $\gamma$ in the equation is the product of the electric charges of the electron and the nucleus (in non-dimensionalized units). Since we are working in one space dimension, we do not have a-priori knowledge of what the physically meaningful range of values is for $\gamma$, so that we treat it as a parameter that can have any positive value.  

Once a numerical solution is found in the interval $[0,S_N]$ for $S_N>0$ suitably large, by reflecting it across the initial point we can obtain a numerical approximation to the desired saddles connector in the interval $[-S_N,S_N]$.  However, since the end points of such a connector are necessarily saddle-nodes, the orbit itself is expected to be highly unstable, and therefore we expect that, when solving the equation forward in $s$, the computed solution will either overshoot or undershoot its target, depending on whether the initial guess for the energy is above or below the actual eigenvalue. Thus by doing a binary search, we can successively halve the length of the interval in which the eigenvalue lies, until it is less than $2\epsilon$.  

We can now investigate the relationship between $\gamma$ and the corresponding energy eigenvalues for various winding numbers (see Fig.~\ref{fig:whit}).


\begin{figure}[H]
    \centering
    \includegraphics[scale=0.5]{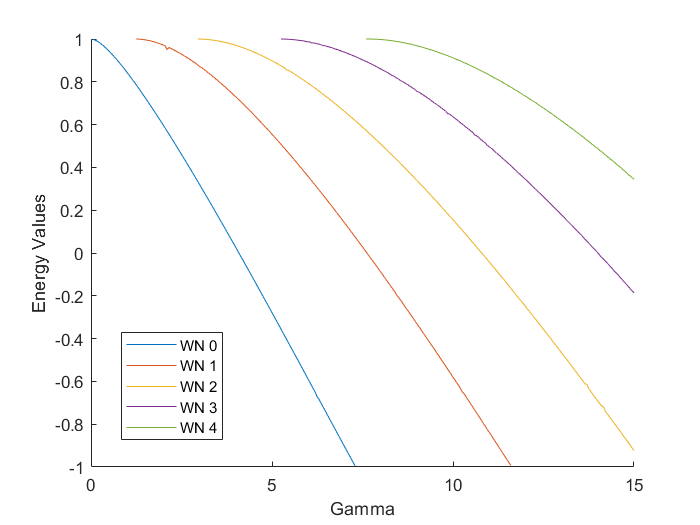}
    \caption{This figure shows the energy eigenvalue as a function of $\gamma$, for different winding numbers.}
    \label{fig:whit}
\end{figure}

From the figure, we can see that only certain winding numbers exist for any given $\gamma$. For example, if we look at $0<\gamma<5$, there only exist winding numbers 0, 1, and 2. The curves of winding numbers 3 and 4 begin at some value $\gamma>5$. This means an atom that corresponds to $\gamma = 5$ has only one ground and two excited states. 

Interestingly, for larger value of $\gamma$, the ground state may not have winding number 0. For example, solutions with three winding numbers exist for $\gamma=7.5$, but this value of $\gamma$ is larger than where the winding number 0 curve crashes into the lower part of the continuous spectrum, so the ground state for this $\gamma$ has winding number 1, and once again, it has two excited states (winding numbers 2 and 3). All of these observations are completely consistent with our Theorem~\ref{Theorem:main}, and are in sharp contrast to the three-dimensional Hydrogenic hamiltonian, where eigenfunctions with any non-negative winding number exist for any value of $\gamma\in(0,1)$, while for $\gamma \geq 1$ the hamiltonian stops being self-adjoint.

For the remainder of this numerical investigation, we will only consider the $\gamma$ values and winding numbers where we expect saddle connectors to exist, according to Figure \ref{fig:whit}.

The instability of saddles connectors that was previously mentioned in the above also means that, in computing $\Theta(s)$, we need to choose $S_N$'s so that they are neither too small (so that the solution has a chance to stabilize) nor too large (so that it has not yet veered off to the wrong equilibrium point at $s=\infty$.) In practice we do this by watching the $\Theta(s)$ values and stopping the computation when we see that the energy stabilizes. See the upper left-hand corner plots in Figures \ref{fig:gamma0.5plots} through \ref{fig:gamma7.5n3plots}.
Once the correct $\Theta(s)$ is found in $[-S_N,S_N]$ we truncate it outside this interval and extend it to all $\mathbb{R}$ by keeping its values constant on each side of that interval (see the upper right-hand plots in Figures \ref{fig:gamma0.5plots} -- \ref{fig:gamma7.5n3plots}.)  We then use this $\Theta(s)$ to compute $R(s)$ by numerically integrating (\ref{eqn:R_theta}),  thus calculating the probability density function $\rho = R^2(s)$, which we plot as a function of $s$ (See the lower left-hand corner plots in  Figures \ref{fig:gamma0.5plots} -- \ref{fig:gamma7.5n3plots}). Finally, we can use equation (\ref{eqn:reduced}) to plot the corresponding eigenfunctions $(u(s),v(s))$ as a parametric curve in the $uv$-plane. These will be curves that have to begin and end at the origin due to the integrability condition \eqref{eqn:normalize}.  See the plots in the lower right-hand corner of figures \ref{fig:gamma0.5plots} -- \ref{fig:gamma7.5n3plots}.  These shapes are therefore the 1-dimensional analogues of the familiar hydrogenic orbitals in 3 dimensional space. 

We first look at the case $\gamma$ = 0.5, where only the ground state ($n=0$) is supposed to exist.
\begin{figure}[H]
   \begin{subfigure}{0.5\textwidth}{\includegraphics[width=0.8\textwidth]{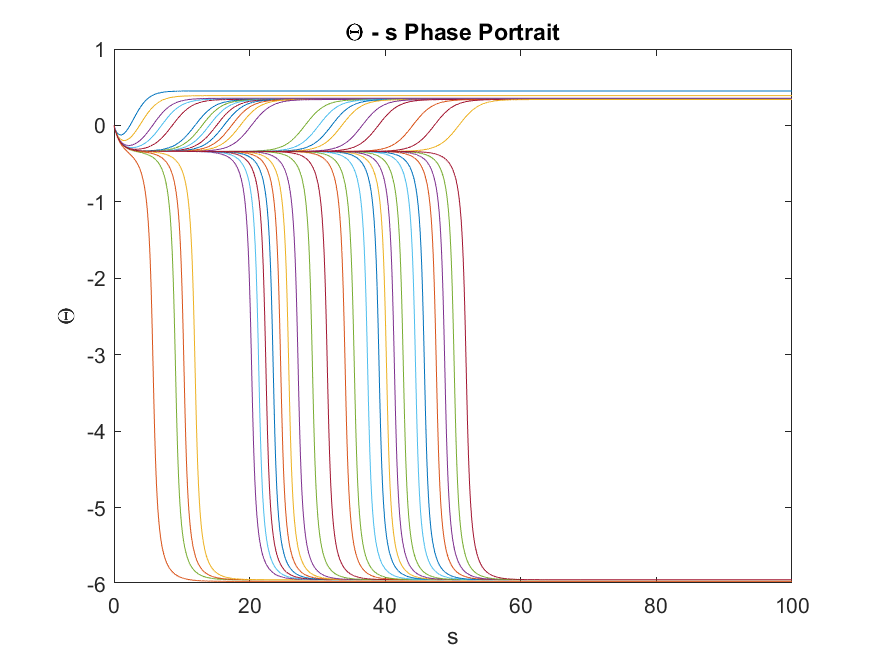}}\end{subfigure}\hfill
 \begin{subfigure}{0.5\textwidth}{\includegraphics[width=0.8\textwidth]{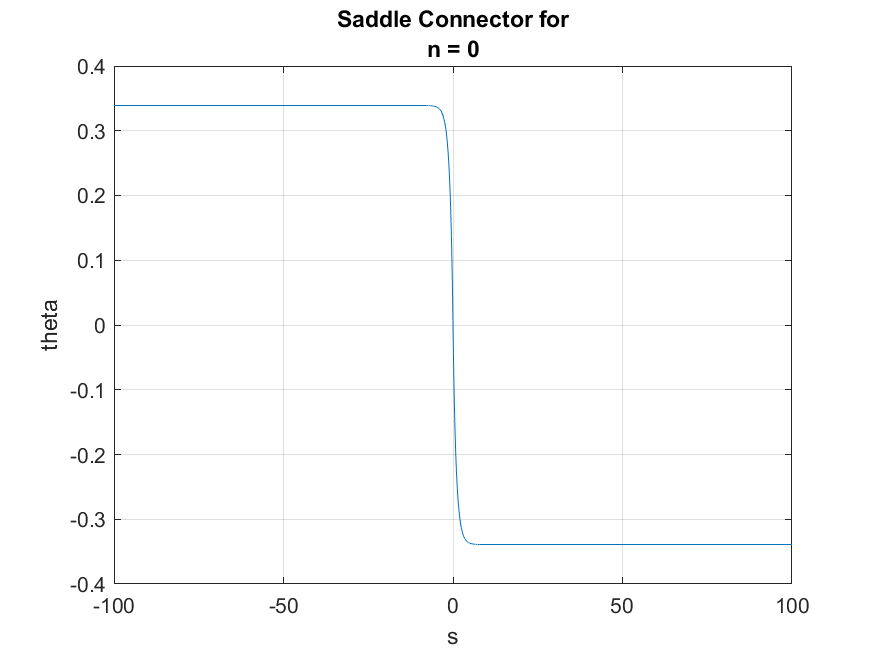}}\end{subfigure}\\
 \begin{subfigure}{0.5\textwidth}{\includegraphics[width=0.8\textwidth]{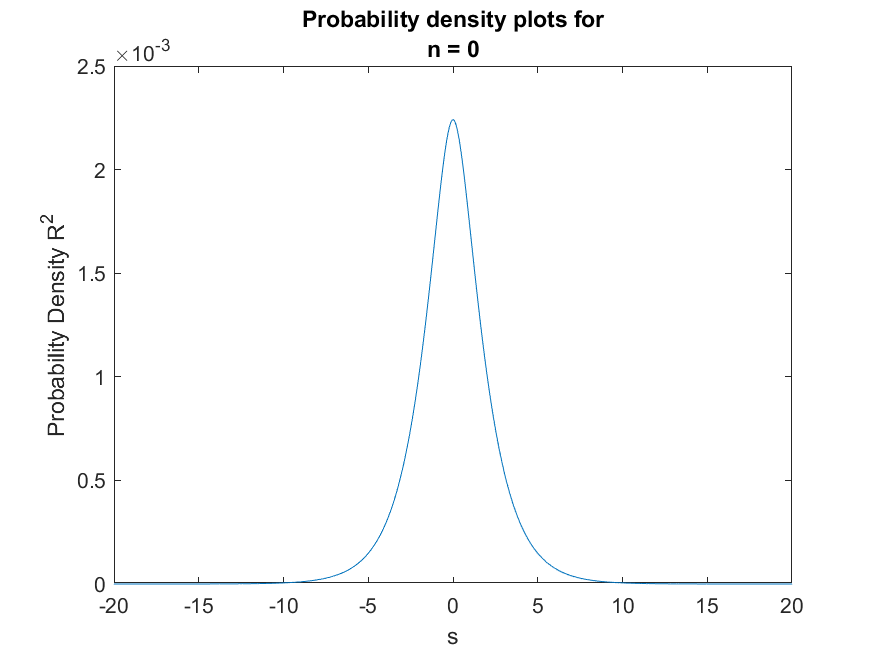}}\end{subfigure}\hfill
 \begin{subfigure}{0.5\textwidth}{\includegraphics[width=0.8\textwidth]{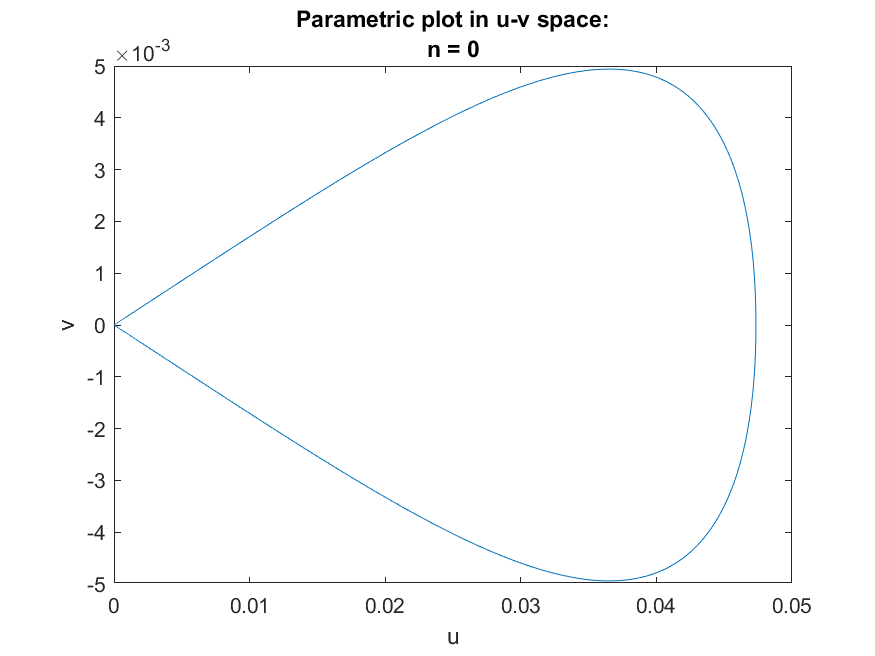}}\end{subfigure}
    \caption{Plots for $\gamma = 0.5$, $n=0$}
    \label{fig:gamma0.5plots}
\end{figure}

The plot on the lower right-hand corner of Figure \ref{fig:gamma0.5plots} shows us that the electron is most likely to be where the nucleus is, and in fact there is a non-zero probability of it being almost on top of the nucleus, in stark contrast to the three-dimensional case. This is to be expected, since unlike the Coulomb potential, our electrostatic potential $\phi$ does not diverge at $s=0$. Another interesting feature of this plot is that it is for the winding number $n=0$ case, and the probability density function has one crest. As we will see in the high $\gamma$ plots, a pattern emerges where the number of crests in the graph of the probability density function is equal to $n+1$, where $n$ is the winding number of the saddles connector in question.

We now look at the case of $\gamma = 7.5$. For this value, the saddles connector with $n=0$ does not exist, and thus the ground state has winding number $n=1$.
\begin{figure}[H]
    \begin{subfigure}{0.45\textwidth}{\includegraphics[width=0.8\textwidth]{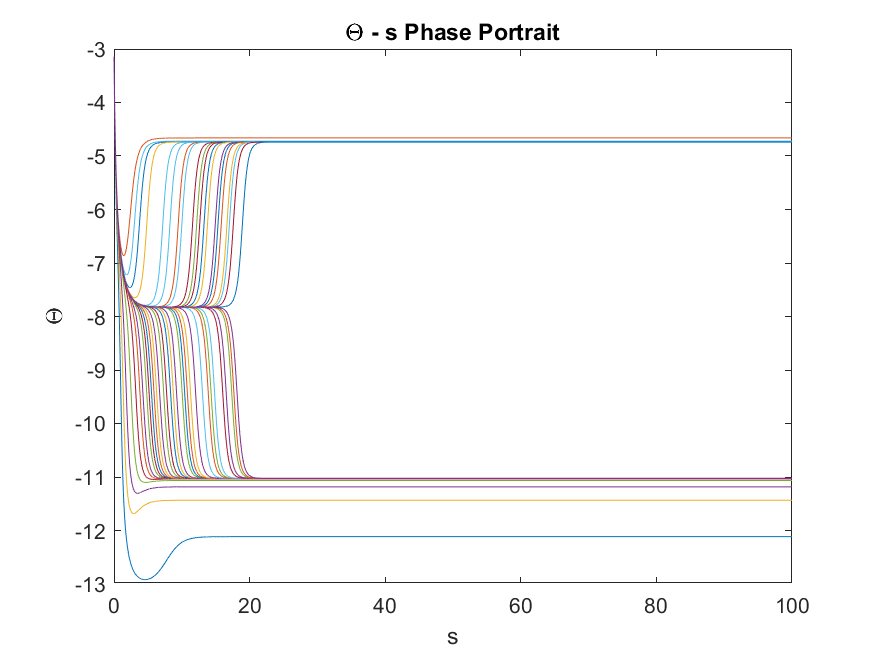}}\end{subfigure}\hfill
     \begin{subfigure}{0.45\textwidth}{\includegraphics[width=0.8\textwidth]{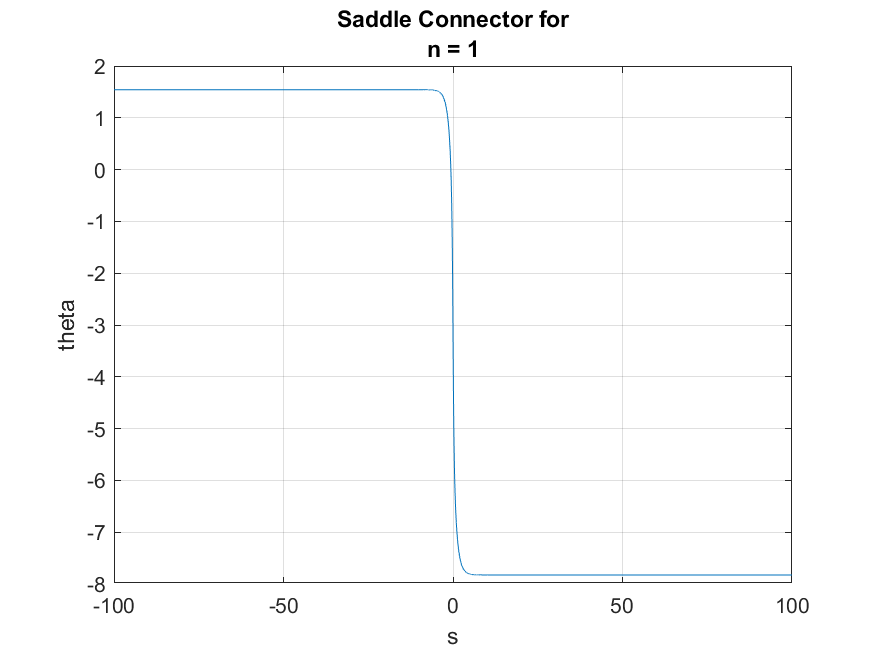}}\end{subfigure}\\
     \begin{subfigure}{0.45\textwidth}{\includegraphics[width=0.8\textwidth]{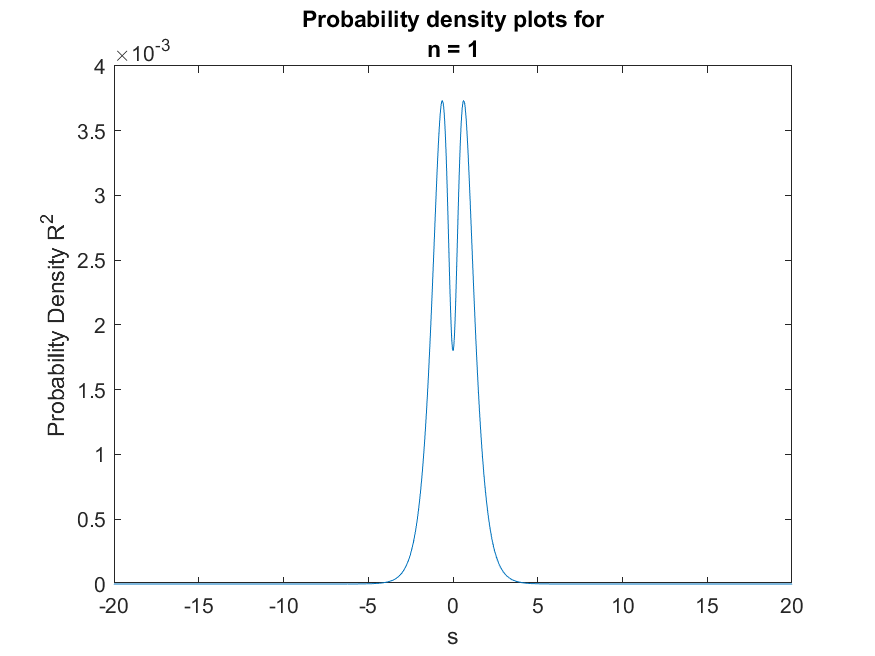}}\end{subfigure}\hfill
     \begin{subfigure}{0.45\textwidth}{\includegraphics[width=0.8\textwidth]{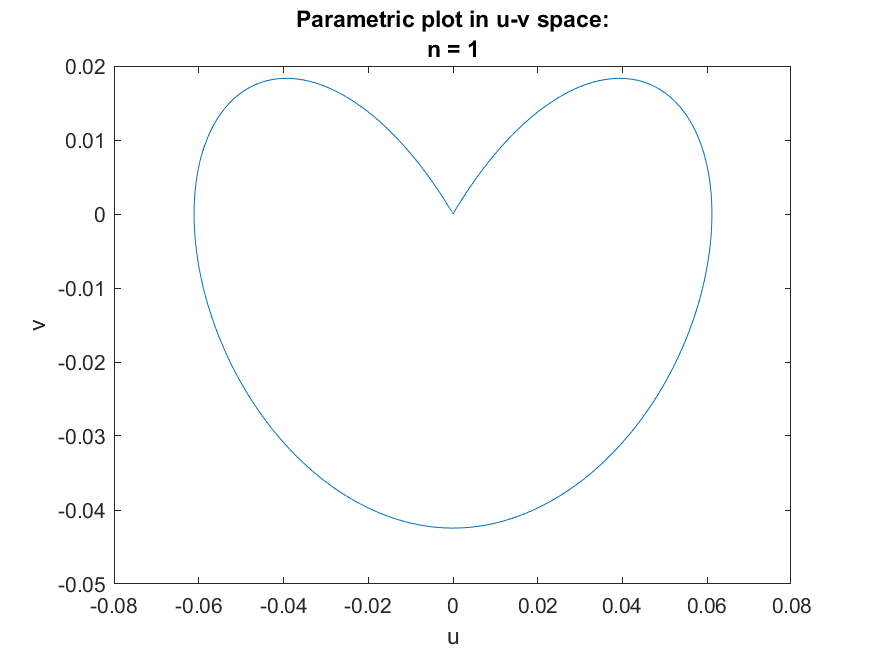}}\end{subfigure}
    \caption{Plots for $\gamma = 7.5$, $n=1$}
    \label{fig:gamma7.5n1plots}
\end{figure}
We can also repeat this analysis for two excited states ($n=2$ and $n=3$). 
\begin{figure}[H]
\begin{subfigure}{0.45\textwidth}{\includegraphics[width=0.8\textwidth]{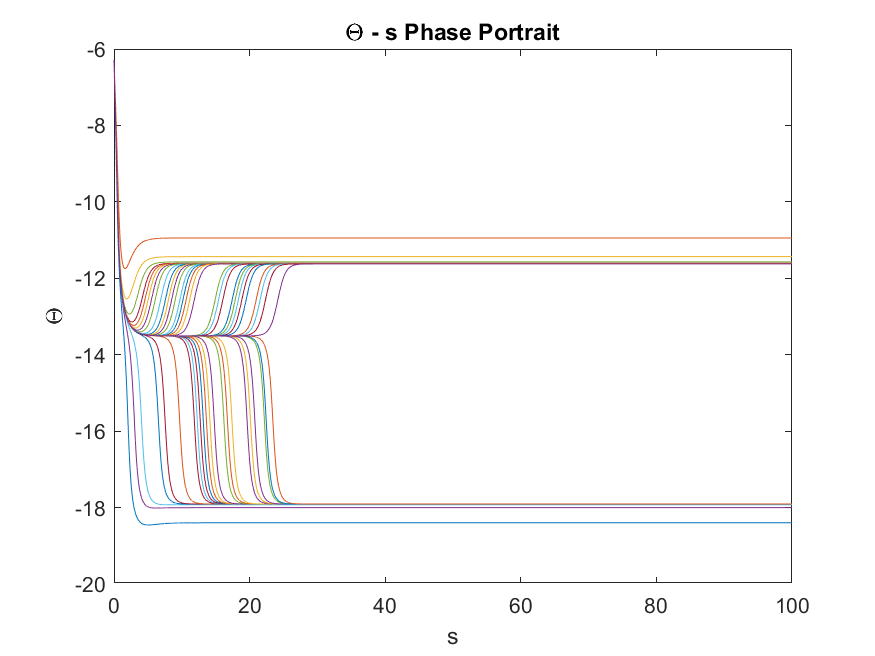}}\end{subfigure}\hfill
\begin{subfigure}{0.45\textwidth}{\includegraphics[width=0.8\textwidth]{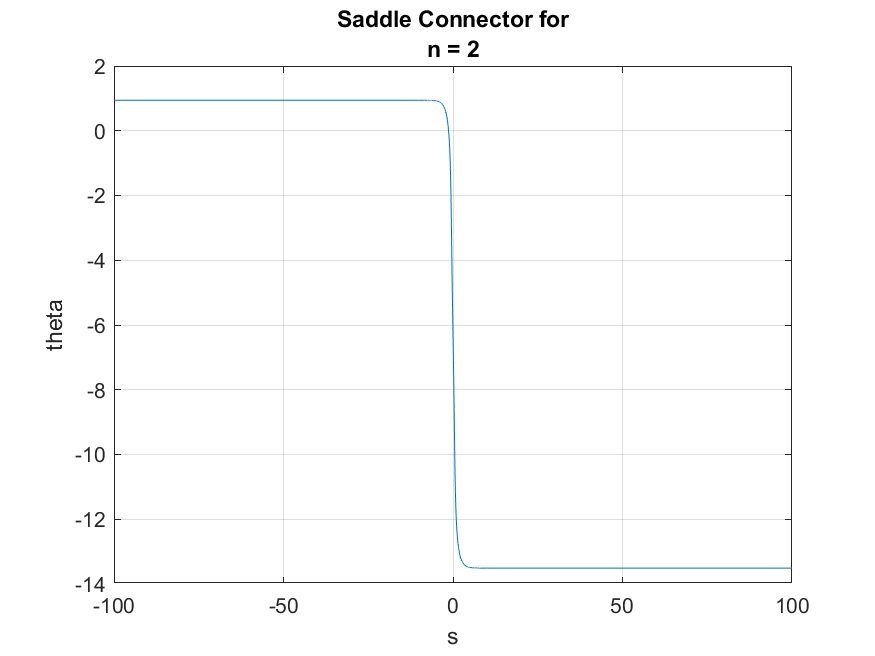}}\end{subfigure}\\
 \begin{subfigure}{0.45\textwidth}{\includegraphics[width=0.8\textwidth]{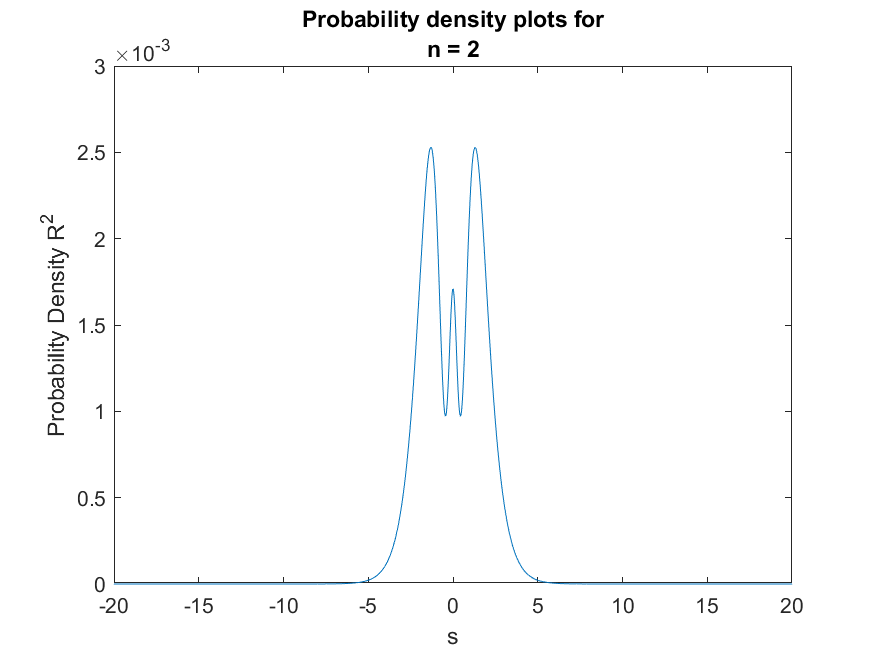}}\end{subfigure}\hfill
\begin{subfigure}{0.45\textwidth}{\includegraphics[width=0.8\textwidth]{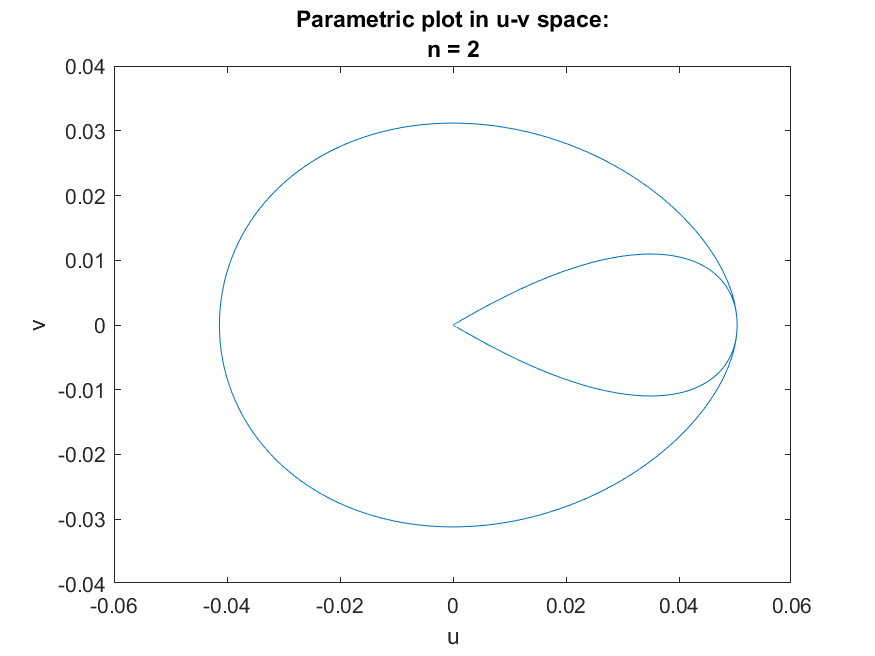}}\end{subfigure}
    \caption{Plots for $\gamma = 7.5$, $n=2$}
    \label{fig:gamma7.5n2plots}
\end{figure}
\begin{figure}[H]
\begin{subfigure}{0.45\textwidth}{\includegraphics[width=0.8\textwidth]{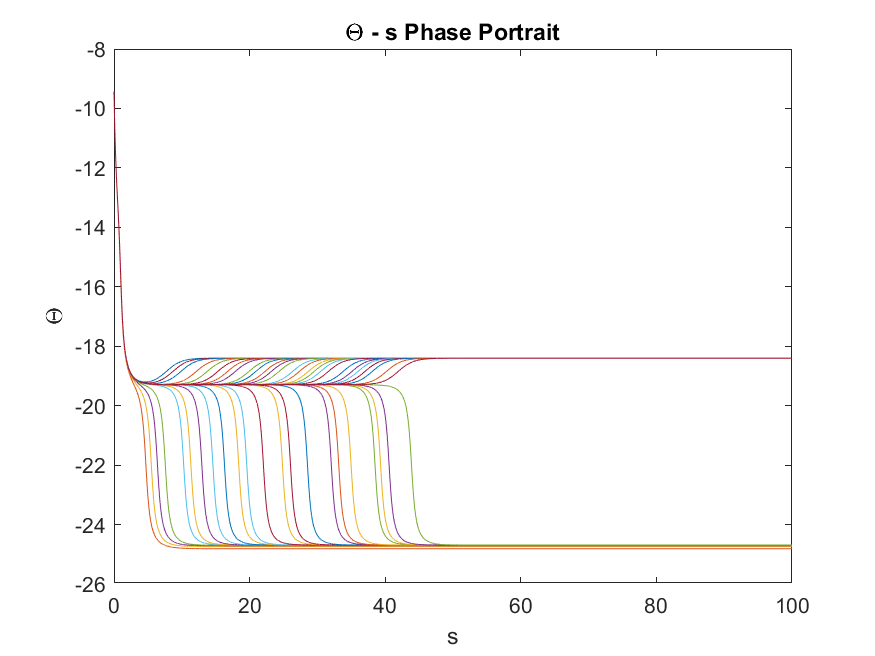}}\end{subfigure}\hfill
\begin{subfigure}{0.45\textwidth}{\includegraphics[width=0.8\textwidth]{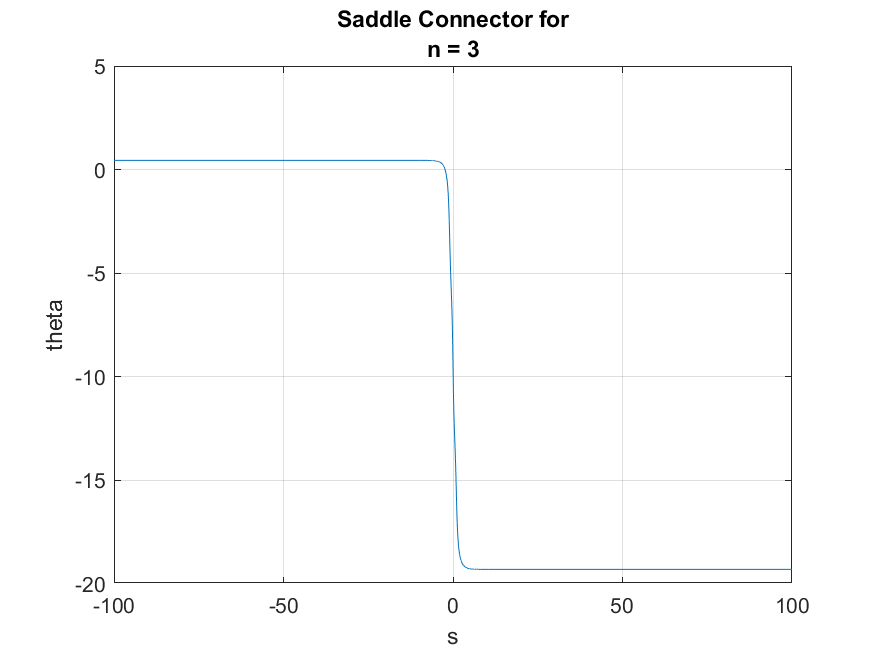}}\end{subfigure}\\
\begin{subfigure}{0.45\textwidth}{\includegraphics[width=0.8\textwidth]{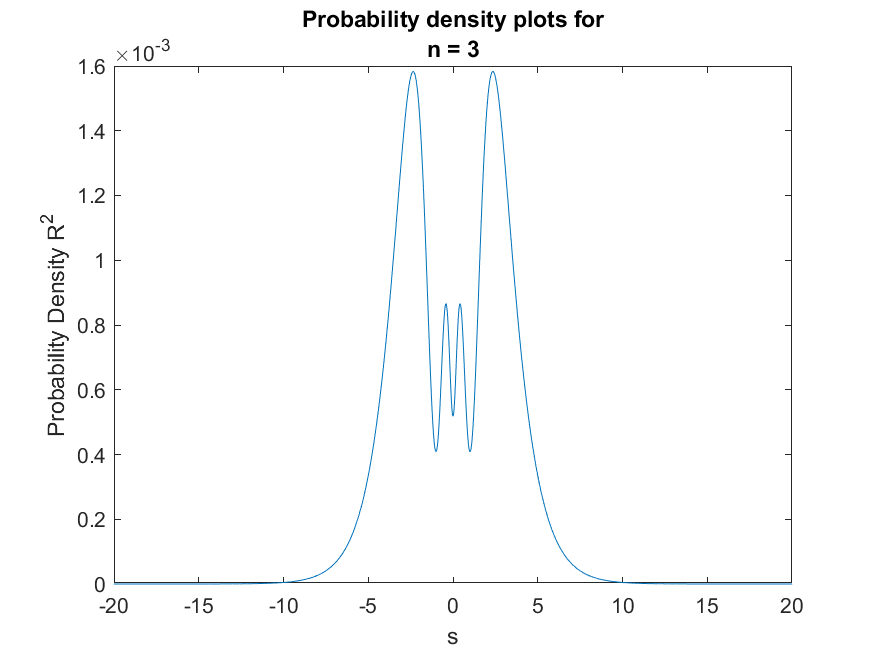}}\end{subfigure}\hfill
\begin{subfigure}{0.45\textwidth}{\includegraphics[width=0.8\textwidth]{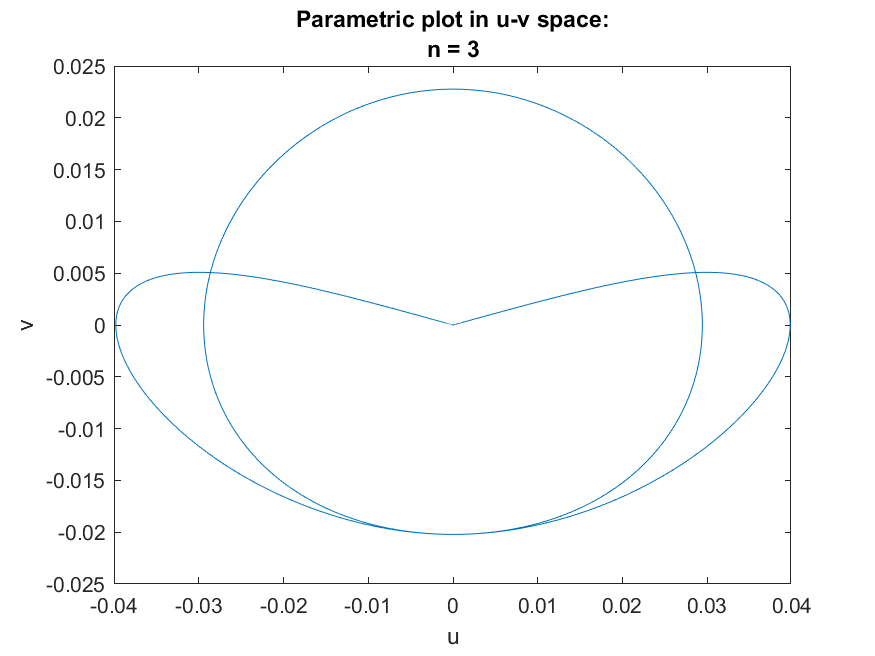}}\end{subfigure}
    \caption{Plots for $\gamma = 7.5$, $n=3$}
    \label{fig:gamma7.5n3plots}
\end{figure}
Interestingly, the graph of $\Theta$ looks very similar for each winding number (as well as in the case $\gamma = 0.5$), except they have different starting and ending points. The probability density plots demonstrate that the electron prefers to stay close to the origin, but there are multiple local maxima (regions where the electron has a higher probability of being located). This is the same conclusion as in the low $\gamma$ case, and these probability density plots have the property that their number of crests equals $n+1$.
\subsection{Zeros and critical points of $M_{\pm i,1/2}$}\label{CWFs}
Using a method discovered by Ikebe \cite{Ike75}, it is possible to efficiently calculate numerical approximations to a large number of zeros and critical points of the Whittaker functions $M_{-i,1/2}$ and $M_{i,1/2}$, or equivalently, for their  Coulomb wave function counterparts.  Here we briefly describe Ikebe's method and how we implemented it in Matlab \cite{MATLAB} in order to compute the constants $\gamma_j$ and $\Gamma_j$ in this paper.

We first recall that the {\em regular Coulomb wave function} $F_{L}(\eta,\rho)$ of order $L = 0,1,2,\dots$ with a real parameter $\eta\in\mathbb{R}$ is by definition the only non-trivial solution of the {\em Coulomb wave equation}
\begin{equation}
    \frac{d^2 w}{d \rho^2} + \left[ 1 - \frac{2 \eta}{\rho} - \frac{L(L+1)}{\rho^2}\right] w = 0,\qquad \rho>0
\end{equation}
that does not have a singularity at $\rho = 0$.  It is well-known (see \cite{NIST:DLMF}, Chap. 33) that
\begin{equation}
    F_{L}\left(\eta,\rho\right)=C_{L,\eta}(- i)^{L+1}M_{i\eta,L+\frac{1}{2}}\left( 2i\rho\right),
\end{equation}
where $M_{\kappa,\mu}(z)$ is the Whittaker $M$-function, and $C_{L,\eta}$ is a normalization constant. Thus in particular, the two Whittaker $M$-functions in this paper correspond to $F_0(\pm 1,2ir)$.  

In \cite{Ike75} Ikebe proved that there exists a symmetric tridiagonal $N\times N$ matrix $T_{L,\eta}$ with the property that $\rho \ne 0$ is a zero of $F_{L}(\eta,\cdot)$ if and only if $\frac{1}{\rho}$ is an eigenvalue of $T_{L,\eta}$, for large enough $N$.  He furthermore showed that $-T_{L,\eta}$ is similar to $T_{L,-\eta}$, which implies that the positive zeros of $F_L(-\eta,\cdot)$ coincide with the absolute values of the negative zeros of $F_L(\eta,\cdot)$.  As a result, by solving an eigenvalue problem for $T_{L,\eta}$ one can simultaneously compute the zeros of $F_L(\eta,\cdot)$ and $F_L(-\eta,\cdot)$. Implementing this eigenvalue problem in Matlab therefore allows us to compute $\gamma_{2j}$ and $\Gamma_{2j}$ simultaneously, for all integers $j\geq 1$. These numerical values were then used to produce the plots in Figure~\ref{fig:staircase}.

Finally, in the same paper Ikebe also showed that the critical points of $F_L(\eta,\cdot)$ can be computed in this way as well, using a slightly different symmetric tridiagonal matrix $\widetilde{T}_{L,\eta}$.  This observation allows us to compute $\gamma_{2j+1}$ and $\Gamma_{2j+1}$ for all integers $j\geq 0$.
\section{Summary and Outlook}\label{sec:sumout}
In this paper we analyzed the spectrum of the Dirac hamiltonian for a single electron in the electrostatic potential of a point nucleus in one spatial dimension and in the Born-Oppenheimer approximation where we fixed the nucleus at the origin.
In order for the discrete spectrum to be non-empty, we had to screen the electrostatic potential so that it had exponential decay at spatial infinity.  We showed that the resulting hamiltonian is essentially self-adjoint and its essential spectrum is the same as the essential spectrum of the Dirac operator in three space dimensions.

To analyze the coupled system of linear ODEs that arise in the study of the discrete spectrum of the hamiltonian, we used a Pr\"ufer transform to recast the equations as a dynamical system on the surface of a finite cylinder. We linearized the system, found the equilibrium points to be exclusively on the two circular boundaries of the cylinder, and determined the local flow near these equilibrium points using center manifold theory, showing that any heteroclinic orbit connecting the two saddle-node points corresponds to a bound state for the electron. We showed that these orbits have a well-defined winding number, and we proved that for any given atomic number for the nucleus, there are only finitely many bound states.

One direction we plan to take to continue this investigation is to find exact formulas or at least good upper and lower estimates for the energy eigenvalues in terms of the other relevant non-dimensional parameters in the problem (nuclear charge, winding number, screening length, etc.)

Another direction we intend to pursue is to couple this electron + nucleus system to a photon, and study the emission/absorption problem in this one-dimensional context.  Do the ground and excited states we found here behave as expected when interacting with a (one-dimensional analog) of a photon field?

A third future task is to incorporate relativistic gravitational effects into this hamiltonian. The above analysis can then be repeated to find the corresponding energy eigenvalues and eigenfunctions in that case.  The difference between the energy eigenvalues in presence of gravity with those in the absence of gravity may have an interpretation as the energy of the one-dimensional analog of gravitons. 


\section{Appendix}
In this appendix we gather certain facts about Whittaker functions and the solutions to the $|E|=1$ boundary value problems that arise in the barrier constructions of \S~\ref{barriers}.
\subsection{Behavior of Whittaker Functions}
Now, consider the series expansion for the Whittaker $M$ function
\begin{equation}
    M_{k,m}(z) = z^{\frac{1}{2}+ m}e^{-\frac{1}{2}z}\sum_{s=0}^\infty \frac{(\frac{1}{2}+m-k)_{s}}{(1+2m)_{s}s!}z^s,
\end{equation}
which is well-defined so long as $2m$ is not a negative integer. In our case, $m=\frac{1}{2}$ and 
$k=-i$. So, using these parameters we get
\begin{equation}
    M_{-i,\frac{1}{2}}(ir) = ire^{-\frac{1}{2}ir}\sum_{s=0}^\infty \frac{(1+i)_{s}}{(2)_{s}s!}(ir)^s = ire^{-\frac{1}{2}ir}(1 + \frac{1+i}{2}ir +O(r^2)),
\end{equation}
and therefore $M_{-i,\frac{1}{2}}(i\cdot0)=0$.

We also compute the derivative
\begin{equation}
    \frac{d}{dr}M_{-i,\frac{1}{2}}(ir) = ire^{-\frac{1}{2}ir}(\frac{1+i}{2}i +O(r))+ (ie^{-\frac{1}{2}ir}+\frac{1}{2}re^{-\frac{1}{2}ir})(1 + \frac{1+i}{2}r +O(r^2)),
\end{equation}
and therefore $\frac{d}{dr}M_{-i,\frac{1}{2}}(i\cdot0)=i$.
It is known that the Whittaker $M$ function will have a zero between consecutive critical points, and similarly, there is a critical point between consecutive zeroes. More exactly, the zeroes and critical points are interlaced and there exist infinitely many of them. The Whittaker $W$ function also has interlaced zeros and critical points and infinitely many of them, which is a known result \cite{u1978ordinary}.

Since $-i\cdot M_{-i,\frac{1}{2}}(i\cdot0)=0$ and $-i\cdot\frac{d}{dr} M_{-i,\frac{1}{2}}(i\cdot0)=1$, the function $-i\cdot M_{-i,\frac{1}{2}}(ir)$ starts at $r=0$ which is its first zero, defined as $\gamma_{0}$, and increases until its first critical point, defined as $\gamma_{1}$. Then,  between $\gamma_1$ (a critical point) and $\gamma_2$ (a zero), $-i\cdot M_{-i,\frac{1}{2}}(z)$ must be decreasing, since $-i\cdot M_{-i,\frac{1}{2}}(\gamma_1)>0$. The sign of the derivative does not change until the next critical point $\gamma_3$ and so the function keeps decreasing until $\gamma_3$. Then, it increases between $\gamma_3$ and $\gamma_5$ (where $\gamma_5$ is the next critical point), with $\gamma_4$ being the zero in between. It must increase, otherwise there would be no zero interlaced between the two critical points. And so, this sine-like behavior continues to repeat indefinitely, since the Whittaker $M$ function is known to have infinitely many zeroes and critical points.  

So, the general pattern is as follows: 
Let $n$ be a non-negative integer

If $r\in(\gamma_j,\gamma_{j+2})$ for $j=4n$, then $-iM_{-i,\frac{1}{2}}(ir)>0$ on that domain. If $j=4n+2$, then $-iM_{-i,\frac{1}{2}}(ir)<0$ on that domain. 

If $r\in(\gamma_j,\gamma_{j+2})$ for $j=4n+3$, then $-i\cdot\frac{d}{dr}M_{-i,\frac{1}{2}}(ir)>0$ on that domain. If $j=4n+2$, then $-i\cdot\frac{d}{dr} M_{-i,\frac{1}{2}}(ir)<0$ on that domain. Also, as mentioned before, $-i\cdot\frac{d}{dr} M_{-i,\frac{1}{2}}(ir)>0$ for $r\in(\gamma_0,\gamma_1)$ and $-i\cdot\frac{d}{dr} M_{-i,\frac{1}{2}}(ir)<0$ for $r\in(\gamma_1,\gamma_3)$.
 
 

\subsection{Zeros and critical points of Whittaker functions}
Recall the definitions \eqref{def:gammaj} and \eqref{def:Gammaj} of the sequence $\gamma_j$ and $\Gamma_j$ respectively. Here we prove the claim that $\gamma_j < \Gamma_j$. We will actually show something slightly stronger, namely,
\begin{lemma}
There exists a constant $r_0>0$ such that    $\gamma_j + r_0 < \Gamma_j$ for $j\geq 1$.
\end{lemma}
\begin{proof}
    We first recall the Sturm-Picone Theorem:  Consider the ODEs
    \begin{eqnarray}
        (p_1(r) y_1')' + q_1(r) y_1 & = & 0,\\
        (p_2(r) y_2')' + q_2(r) y_2 & = & 0.
    \end{eqnarray}
Assume the functions $p_1, p_2, q_1, q_2$ are continuous on an interval $[a,b]$ and that they satisfy
$$ 0 < p_2 \leq p_1,\qquad q_1 \leq q_2.$$
Suppose there are $z_1,z_2 \in [a,b]$, $z_1<z_2$, such that $y_1(z_1) = y_1(z_2) = 0$, and suppose that $y_1$ and $y_2$ are linearly independent.  Then there exists $x \in (z_1,z_2)$ where $y_2(x) = 0$.

We will be apply this theorem to the following two ODEs
\begin{eqnarray}
y_1'' + (\frac{1}{4} - \frac{1}{r}) y_1 = 0,\label{eq:whit1}\\
y_2'' + (\frac{1}{4} + \frac{1}{r-r_0}) y_2 = 0,\label{eq:whit2}
    \end{eqnarray}
for some $r_0>0$ to be determined later. Thus we have
$$
p_1 = p_2 = 1, \qquad q_1 = \frac{1}{4} - \frac{1}{r},\qquad q_2 = \frac{1}{4} + \frac{1}{r-r_0}.
$$
Let $N>0$ be a fixed large number.  On the interval $[r_0+1,N]$ all the hypotheses of the Sturm-Picone theorem are satisfied.  Moreover, \eqref{eq:whit1} is the Whittaker equation for $\kappa=i, \mu= \frac{1}{2}$, so we know that the zeros of $y_1$ are the sequence $\{\Gamma_{2j}\}_{j\in\mathbb{N}}$.  It thus follows that for all integers $j\geq 2$ there exists $x \in (\Gamma_{2j-2},\Gamma_{2j})$ such that $y_2(x) = 0$. Let $x_j$ denote the largest such $x$ in that interval.   Now, \eqref{eq:whit2} is just the Whittaker equation for $\kappa=-i, \mu = \frac{1}{2}$ with the independent variable shifted by the amount $r_0$, thus we have $y_2(\gamma_{2k}+r_0) = 0$ so that we must have $x_j = \gamma_{2k} + r_0$ for some $k$.  We need to show that $k\geq j$. We proceed by induction on $j$:  For $j=2$ we know $\Gamma_2 \approx 11.6282$ and $\Gamma_4 \approx 18.9491$, while $\gamma_4 \approx 7.6437$. Thus for any $r_0 \in [4,11]$ we have $\Gamma_2 < \gamma_4+r_0 < \Gamma_4$, therefore $k=2=j$.  

Now assume the induction step, i.e. $\Gamma_{2j-2} < x_j = \gamma_{2k}+r_0 < \Gamma_{2j}$ with $k\geq j$, and $x_j$ is the largest zero of $y_2$ in the interval $(\Gamma_{2j-2},\Gamma_{2j})$.  This means that the next zero of $y_2$ must be greater than $\Gamma_{2j}$, i.e. $\gamma_{2k+2} +r_0 > \Gamma_{2j}$.  Let $k'$ be defined by $x_{j+1} = \gamma_{2k'}+r_0$.  Since $x_{j+1}$ is the largest zero of $y_2$ in $(\Gamma_{2j},\Gamma_{2j+2})$, we must have $x_{j+1} \geq \gamma_{2k+2} + r_0$. Hence $k'\geq k+1 \geq j+1$ and the induction argument proceeds, showing that $\gamma_{2j}+ r_0 < \Gamma_{2j}$ for all $j\geq 2$ such that $\Gamma_{2j} \leq N$.  Since $N$ was arbitrary, this established the claim for all $j\geq 2$. The remaining case, $j=1$ is evident from the numerical approximations $\gamma_2 \approx 2.9347$ and $\Gamma_2 \approx 11.6282$.  We therefore also have $\gamma_2 + r_0 < \Gamma_2$ for $r_0 \leq  8.69$. 

This takes care of the statement about the zeros.  To prove the corresponding statement for the critical points, we apply the Sturm-Picone theorem again, this time to the equations satisfied by the {\em derivatives} of functions $y_1,y_2$ that solve \eqref{eq:whit1} and \eqref{eq:whit2}.  The equations satisfied by the derivatives are
\begin{eqnarray}
    \left( \frac{r}{r+4} \eta_1' \right)' + \frac{1}{4}\eta_1 & = & 0, \label{eq:whitder1}\\
    \left( \frac{r-r_0}{r-r_0-4} \eta_2' \right)' + \frac{1}{4}\eta_2 & = & 0,\label{eq:whitder2}
\end{eqnarray}
where we have once again shifted the independent variable in the second equation by an amount to be determined, $r_0>0.$ We thus have
$$ p_1 = \frac{r}{r+4},\qquad p_2 = \frac{r-r_0}{r-r_0-4},\qquad q_1 = q_2 = \frac{1}{4}.$$
On the interval $[r_0+5,N]$, with $N$ large as before, we have that the hypotheses of Sturm-Picone are once again satisfied. Since \eqref{eq:whitder1} is the equation satisfied by the first derivative of the Whittaker function $M_{i,1/2}$, we know 
$\eta_1(\Gamma_{2j-1}) = \eta_1(\Gamma_{2j+1}) = 0$ for all $j\geq 1$, and therefore there exists $\xi_j \in (\Gamma_{2j-1},\Gamma_{2j+1})$ such that $\eta_2(\xi_j) = 0$. We again let $\xi_j$ be the largest number with these properties.  Now, because \eqref{eq:whitder2} is the equation satisfied by the first derivative of $M_{-i,1/2}$ (shifted by $r_0$), we have that $\xi_j = \gamma_{2k+1} + r_0$ for some $k$.  Once again we can use induction to prove that $k \geq j$, and by letting $N \to \infty$ this establishes the claim $\gamma_{2j+1} + r_0 < \Gamma_{2j+1}$ for all $j\geq 1$.  The remaining case $j=0$ can be checked by hand: $\gamma_1 \approx 1.2309$, $\Gamma_1 \approx 7.3148$, so that $\gamma_1 + r_0 < \Gamma_1$ for $r_0 \leq 6.08$. Finally, we have made the tacit assumption that $\Gamma_1 > r_0 +5$, which requires $r_0 \leq 2.31$.  This complete proof of the statement, and we can take $r_0 = 2.31$.
\end{proof}

\subsection{Zeros and critical points of the solutions to the $|E|=1$ boundary value problems} \label{CritProof}

\begin{proposition}\label{prop:crits}
    The solutions of the boundary value problems \eqref{eq:bvpodd} and \eqref{eq:bvpeven} have one and only one critical point in the interval $(\gamma_0,\gamma_1)$.
\end{proposition}
\begin{proof}
 We evaluate $\dot{w}(r)$ by first substituting equations 13.14.2	and 13.14.3 of \cite{NIST:DLMF} for for  $M_{-i,1/2}(ir)$ and  $W_{-i,1/2}(ir)$ respectively into \eqref{eq:bvpodd} and \eqref{eq:bvpeven}. Then we differentiate. We evaluate the subsequent $M'_{-i,1/2}(ir)$ and $W'_{-i,1/2}(ir)$ terms separately and show that first is finite while the latter diverges to $\infty$ at $r=0$. 

 The definition of the Whittaker M function as given by \cite{NIST:DLMF} is 
 \begin{equation}
     M_{-i,\frac{1}{2}}(z) = e^{-\frac{z}{2}}z^{(\frac{1}{2}+\mu)}M(\frac{1}{2}+\mu-\kappa,1+2\mu,z).
 \end{equation}
In our case, $\mu = \frac{1}{2}$ and $\kappa = -i$. If we define
\begin{equation}
    f(z) = e^{-\frac{z}{2}}z^{(\frac{1}{2}+\mu)},
\end{equation}
it is clear that $f(0) = 0$ and $f'(0) = i$. So, in order to show $M'_{-i,1/2}(ir)$ is finite at $r=0$, it will suffice to show that $M(1+i,2,z)$ and $M'(1+i,2,z)$ is finite at $z=0$. As per 13.2.2 of \cite{NIST:DLMF}
\begin{equation}
M\left(a,b,z\right)=\sum_{s=0}^{\infty}\frac{{\left(a\right)_{s}}}{{\left(b%
\right)_{s}}s!}z^{s}=1+\frac{a}{b}z+\frac{a(a+1)}{b(b+1)2!}z^{2}+\cdots.
\end{equation}
In our case, $a = 1+i$ and $b=2$ so
\begin{equation}
M'\left(a,b,0\right)= \frac{a}{b} = \frac{1+i}{2}.
\end{equation}
This is finite, and the $M'_{-i,1/2}(i\gamma)$ in our solution to the odd $N$ boundary value problem \eqref{eq:oddsol} will also have a finite coefficient for the first term so long as $\gamma \neq \gamma_j$ for odd $j$ and $j=0$. Similarly, the first term of the even $N$ boundary value problem \eqref{eq:evensol} will be finite so long as $\gamma \neq \gamma_j$ for even $j$. This condition is necessary for the respective boundary value problems. 

The definition of the Whittaker $W$ function given by \cite{NIST:DLMF} is 
\begin{equation}
W_{\kappa,\mu}\left(z\right)=e^{-\frac{1}{2}z}z^{\frac{1}{2}+\mu}U\left(\tfrac%
{1}{2}+\mu-\kappa,1+2\mu,z\right),
\end{equation}
where 
\begin{multline}
    U\left(a,n+1,z\right)=\frac{(-1)^{n+1}}{n!\Gamma\left(a-n\right)}\sum_{k=0}^{%
\infty}\frac{{\left(a\right)_{k}}}{{\left(n+1\right)_{k}}k!}z^{k}\left(\ln z+%
\psi\left(a+k\right)-\psi\left(1+k\right)-\psi\left(n+k+1\right)\right) \\
+\frac{%
1}{\Gamma\left(a\right)}\sum_{k=1}^{n}\frac{(k-1)!{\left(1-a+k\right)_{n-k}}}{%
(n-k)!}z^{-k}
\end{multline}
is the Kummer $U$ function and $\psi$ is the digamma function. 

Since the derivative of the $M'_{-i,1/2}(0)$ term is finite, it will suffice to show that the real part of $\Gamma(1+i)\frac{d}{dr}W_{-i,1/2}{(ir)}$ is infinitely positive at $r=0$. The imaginary parts of the $M'$ and $W'$ terms will cancel since the derivative of our solution must also be real. 

In our case, 
\begin{equation}\label{ddrW}
    \Gamma(1+i)\frac{d}{dr}W_{-i,1/2}(ir)= i\Gamma(1+i)(f(ir)U'{(1+i, 2, ir)} + f'{(ir)}U{(1+i, 2, ir)}).
\end{equation}
By using our formula for $U(1+i,2,ir)$, we determine that the real part limit of this equation as $r$ approaches $0$ is equivalent to 
\begin{equation}
\lim _{r \to 0 }i\frac{\Gamma(1+i)}{\Gamma(i)}\ln{(r)} = -\lim _{r \to 0 } \ln(r) = \infty.
\label{eq:limit}
\end{equation}
Thus, $w'(\gamma_0) >0 $.

In the $E=-1$ case, it is simple to check that when computing the derivative of $u(r)$ at $\Gamma_0$, the limit in equation \ref{eq:limit} would be $\lim _{r \to 0 }i\frac{\Gamma(1-i)}{\Gamma(-i)}\ln{(r)} = \lim _{r \to 0 } \ln(r) = -\infty$ instead.


 On the other hand, $w'(\gamma_1) < 0$. $\gamma_1$ is a root of $M'_{-i,\frac{1}{2}}(ir)$, so this term is 0. The $W'_{-i,\frac{1}{2}}(\gamma_1)$ term is approximately $-1.8$, so the derivative at $\gamma_1$ for our solution is always negative. Then, since $w'(\gamma_0) >0 $ and $w'(\gamma_1) <0 $, there must be a critical point of $w$ between $\gamma_0$ and $\gamma_1$. We claim this will be the only critical point in this interval. Suppose there were two critical points between $(\gamma_0, \gamma_1)$. By Rolle's theorem, there will be a point $r_0$ in between where the second derivative is zero, and by our differential equation, this point is a root of $w(r)$. That requires the sign of $w(r)$ must change at $r_0$, since if it did not, then there would be at least three critical points in the interval or in other words two points at which $w''(r)=0$ by Rolle's theorem or two roots between $(\gamma_0,\gamma_2)$. This violates Sturm's separation theorem which guarantees at most one root in this interval between consecutive roots of $M_{-i,1/2}(ir)$ \cite{u1978ordinary}. 
 So the other possibility is that the sign of $w(r)$ changes. However, by Theorem 1 of \cite{beesack1972sturm}, if $y_2$ and $y_1$ are independent solutions of our differential equation on an open interval $(a,b)$, $\frac{y_2}{y_1}$ will be monotonic on that open interval. $M_{-i,\frac{1}{2}}(ir)$ is one such solution, and it must be monotonic between $\gamma_0$ and $\gamma_1$. Using Theorem 1 of \cite{beesack1972sturm}, this implies that the quotient of linearly independent solutions $w(r)/M_{-i,\frac{1}{2}}(ir)$ must be monotonic. Then, $w(r)$ is monotonic under the condition that $w'(r)$ or $M'_{-i,\frac{1}{2}}(r)$ are not flat zero for on an interval. This condition holds, for if it didn't,  this would require three or more critical points due to Rolle's theorem and this would violate Sturm's separation theorem since there would be two or more roots in the interval. So, the zeros are isolated. This guarantees that both $w(r)$ and $M_{-i,\frac{1}{2}}(r)$ must remain monotonic in order to satisfy their division being monotonic. Hence, $w(r)$ is monotonic on $(\gamma_0,\gamma_1)$. This contradicts the fact that the sign of $w(r)$ would have to change at $r_0$. Hence, our hypothesis that there were two critical points is false.  Thus, if $\gamma \in (\gamma_0,\gamma_1)$ there is exactly one critical point of $w(r)$.

 These results apply to both the odd and even winding number solutions \eqref{eq:oddsol} and \eqref{eq:evensol}.  
\end{proof}

The above is used in the proof of the following:
\begin{lemma}\label{lem:Lem1}
Let $\gamma \in [\gamma_{k-1},\gamma_k) \cap [\Gamma_{j-1},\Gamma_j) $. For the $E=1$ boundary value problems, the general solution $w(r)$ to \eqref{eq:bvpodd}
 will have $\lfloor \frac{k}{2}+1 \rfloor$ critical points in $r\in(0,\gamma]$ and the general solution to \eqref{eq:bvpeven} will have $\lfloor \frac{k}{2}+\frac{1}{2} \rfloor$ critical points.  For the boundary value problems corresponding to $E=-1$, the general solution to \eqref{eq:bvpodd2} will have  $\lfloor \frac{j}{2} \rfloor$ roots and the general solution to \eqref{eq:bvpeven2} will have  $\lfloor \frac{j}{2}+\frac{1}{2} \rfloor$.
 \end{lemma}

 \begin{proof}
     We use induction and start with the $E=1$ cases. 
     
     We define
    \begin{equation}
        f(r) = c_{1}M_{-i,1/2}(ir).
    \end{equation}
    
    We begin by differentiating the boundary value problem \eqref{eq:bvpodd} for odd winding number solutions with respect to $r$ once more. We get 
    \begin{equation}
      \dddot{w} + \left( \frac{1}{4} + \frac{1}{r} \right) \dot{w} +\frac{1}{r^2}\frac{\ddot{w}}{\frac{1}{4}+\frac{1}{r}} = 0,
      \label{eq:thirdderiv}
    \end{equation}
    where we substitute $w = -\frac{\ddot{w}}{\frac{1}{4}+\frac{1}{r}}$ given by our original differential equation. Then, we rewrite this with $z=\dot{w}$ as
    \begin{equation}\label{zddot}
      \ddot{z} + \left( \frac{1}{4} + \frac{1}{r} \right) z +\frac{1}{r^2}\frac{\dot{z}}{\frac{1}{4}+\frac{1}{r}} = 0.
    \end{equation}
     This is a homogeneous second order linear differential equation of $z$. As such, we can apply the Sturm separation theorem \cite{beesack1972sturm}. Both $z = \dot{w}(r)$, where $w(r)$ is given by the odd boundary value problem solution \eqref{eq:oddsol}, and $\dot{f}(r)$ satisfy this differential equation (for different initial value problems). By the Sturm separation theorem, $z(r)=\dot{w}(r)$ has exactly one root between successive roots of $\dot{f}(r) = i c_{1}M'_{-i,1/2}(ir)$. In other words, there is exactly one critical point of $w$ between successive roots of $\dot{f}(r)$ which are given by $(\gamma_j, \gamma_{j+2})$ for $j$ odd. Similarly, between successive roots of $f(r)$, corresponding to $\gamma$ with even indices, there is only one root to our general solution to both boundary value problems. This applies to both the $E=1$ and $E=-1$ cases, one only needs to substitute $r$ with $-r$ and $\gamma_i$ with $\Gamma_i$ and all else holds. 
         
     Let $k=1$, then $\gamma\in(\gamma_0,\gamma_1)$. In the case of the odd boundary value problem \eqref{eq:bvpodd}, we are guaranteed one critical point due to the boundary value condition. There cannot be any more, as proven in Proposition~\ref{prop:crits}. Similarly, the even boundary value problem \eqref{eq:bvpeven} gives us one critical point because $\dot{w}(\gamma_0) > 0$ and $w(\gamma) = 0$, and more than one critical point would give us more than one zero in this interval due to Rolle's theorem and our differential equation, which would violate Sturm's theorem. Indeed, for $k=1$, there are $\lfloor \frac{1}{2}+1 \rfloor =1$ critical point for the solution to \eqref{eq:bvpodd} and $\lfloor \frac{1}{2}+\frac{1}{2} \rfloor =1$ critical points for the solution to \eqref{eq:bvpeven}.

     Now suppose that for $k = n$, the number of critical points is  $\lfloor \frac{n}{2}+1 \rfloor$ for the solution to \eqref{eq:bvpodd} and  $\lfloor \frac{n}{2}+\frac{1}{2} \rfloor$ for the solution to \eqref{eq:bvpeven}.

     Now let $k= n+1$. We claim that if $k$ is even, the solution of \eqref{eq:bvpodd} will gain one additional critical point and that for $k$ odd, the solution of \eqref{eq:bvpeven} will gain one additional critical point.  We have already shown in Proposition~\ref{prop:crits} that the critical point between $(\gamma_0,\gamma_1)$ will be preserved for all $\gamma$. Furthermore, the Sturm critical points between consecutive odd $\gamma_i$ will also be preserved. If $k$ is even, then we add a new Sturm critical point between $(\gamma_{k-3},\gamma_{k-1})$ to the solution of \eqref{eq:bvpodd}, in addition to our boundary value condition. Otherwise, if $k$ is odd, the previous critical points that are guaranteed in the $k=n$ case are preserved and no new critical points can enter due to Sturm's separation theorem. Similarly, for $k$ odd, the critical points of \eqref{eq:bvpeven} are preserved, and there will be an additional critical point between $(\gamma_{k-1},\gamma_k)$ due to the addition of a Sturm zero between $(\gamma_{k-3},\gamma_{k-1})$ and the zero at our boundary value.  This new critical point must occur in this particular interval because between intervals of consecutive even-indexed $\gamma_i$, roots must come after critical points, as this is the starting behavior of our solution for all $\gamma$, and due to the alternation of critical points and roots, violating this behavior at larger values of $r$ would result in one too many roots in at least one such interval.

     Counting up the critical points in each case, we verify that there are $\lfloor \frac{k}{2}+1 \rfloor$ for the solution to \eqref{eq:bvpodd} and  $\lfloor \frac{k}{2}+\frac{1}{2} \rfloor$ for the solution to \eqref{eq:bvpeven}, proving our inductive hypothesis. 
     
     The $E=-1$ case is similar in proof. The base case is shown using the boundary value conditions of (\ref{eq:bvpodd2}) and (\ref{eq:bvpeven2}). For $\Gamma\in(\Gamma_0,\Gamma_1)$, we are only guaranteed a zero in the case of \eqref{eq:bvpeven2} due to the boundary value condition, and there is only one such zero due to Sturm's separation theorem. The case of \eqref{eq:bvpodd2} does not have a Sturm zero for this range of $\Gamma$. To prove that there is no root at all, we examine the solution to our equation \eqref{eq:oddsol2} on $(\Gamma_0,\Gamma_1)$. 
      The second term is a multiple of the Whittaker-W function whose first root for $r>0$ is at $r \approx  7.55 > \Gamma_1$ and is outside the interval. The function then must be decreasing inside the interval since the derivative at $r=0$ is negative and the first critical point is outside the interval. At $\Gamma_0=0$, the first term is zero and by our boundary value condition, we have that $\Gamma(1-i)W_{i,1/2}(i\cdot 0) = 1$ which is positive. Hence, the real part of $\Gamma(1-i)W_{i,1/2}(ir)$ is positive between $(\Gamma_0,\Gamma_1)$. The real part of the first term is monotonic because it is some multiple $M_{i,1/2}(ir)$ and we are in the open interval of a root followed by the next critical point of the function. For $\gamma \in (\Gamma_0,\Gamma_1)$, the real part of the quotient factor $\frac{W'_{i,1/2}(i\gamma)}{M'_{i,1/2}(i\gamma)}$ is negative solely in $2.11 \lesssim \gamma \lesssim 6.31$ and obtains a minimum value of $-0.00671$ at $\gamma=4$. The maximum real value of $-\Gamma(1-i)M_{i,1/2}(ir)$ in this negative interval is at $r=6.31$ due to the monotonicity. This is because the real part of this function is increasing which can be verified by computing the derivative at any point in this interval. The value of this maximum is  $\mathfrak{Re}[-\Gamma(1-i)M_{i,1/2}(i \cdot 6.31)] \approx 3.22$. So, the product of the real parts of each factor in the first term is at least $3.22\cdot-0.00671\approx -0.022$. However, $\Gamma(1-i)W_{i,1/2}(ir)$ is decreasing on the interval (its first root for $r>0$ is outside the interval and we know that this function is $1$ at $r=0$), and its positive real part at $r=6.31$ is $0.055$. Since $0.055-0.022=0.033>0$, we can be assured that the product of the real parts of the factors, when summed with the real part of the second term, will never be negative in $r\in(\Gamma_0,\Gamma_1)$ for any value of $\gamma \in (\Gamma_0,\Gamma_1)$. However, we must also consider the real value of the product of imaginary parts in the first term. 
      The imaginary part of the quotient factor has a positive minimum at $\gamma=4$ on this interval, and the imaginary part of $-\Gamma(1-i)M_{i,1/2}(ir)$ is negative and decreasing, which can be verified by simply computing the derivative at $r=0$ and using the fact that the function is monotone on the interval. So, the product of the imaginary parts for the first term will have a positive real part contribution. 

      Hence, we can conclude that $u(r)$ will be positive on the interval and therefore has no root in $(\Gamma_0,\Gamma_1)$ for $\gamma \in (\Gamma_0,\Gamma_1)$ in the case of (\ref{eq:bvpodd2}).  This establishes the base case, since for $j=1$, there are $\lfloor \frac{1}{2} \rfloor = 0$ roots for the solution to \eqref{eq:bvpodd2} and $\lfloor \frac{1}{2}+\frac{1}{2} \rfloor =1$ root for the solution to \eqref{eq:bvpeven2}.

     Now suppose that for $j = n$, the number of roots is  $\lfloor \frac{n}{2} \rfloor$ for the solution to \eqref{eq:bvpodd2} and  $\lfloor \frac{n}{2}+\frac{1}{2} \rfloor$ for the solution to \eqref{eq:bvpeven2}.
     
     Now let $j = n+1$. If $j$ is odd, then  the solution to \eqref{eq:bvpeven2} preserves the number of Sturm roots from the previous intervals of consecutive, even-indexed $\Gamma_i$ that existed in the $j=n$ case, although not necessarily at the exact same value, and gains one new root between $(\Gamma_{j-1},\Gamma_{j})$ due to the boundary condition. The solution does not gain an additional root if $j$ is even since there can only be one root in $(\Gamma_{j-2},\Gamma_{j})$ which had already been fulfilled by our boundary value condition. For the solution to \eqref{eq:bvpodd2}, it will also preserve all Sturm roots from the previous even-indexed intervals of $\Gamma_i$. In addition to this, it will preserve all Sturm critical points in the odd-indexed intervals. If $j$ is even, then a new Sturm critical point will be added to the existing ones from the $j=n$ case between $(\Gamma_{j-1},\Gamma_j)$. This will induce a new root between the last two critical points, and so the solution to \eqref{eq:bvpodd2} will gain an additional root if $j$ is even. On the other hand, if $j$ is odd, there is no new critical point, and no new root will be added. If another root were to be added, then it would violate Sturm's separation theorem in one of the intervals, which one can verify using Rolle's theorem and the fact that the order of critical points and roots must be preserved in each of the even-indexed intervals of $(\Gamma_i)$, an argument analogous to the $E=1$ case. 

     Then, we count up the the total number of roots and find that the general solution to \eqref{eq:bvpodd2} will have  $\lfloor \frac{j}{2} \rfloor$ roots and the general solution to \eqref{eq:bvpeven2} will have  $\lfloor \frac{j}{2}+\frac{1}{2} \rfloor$ roots, which proves our inductive hypothesis. 
     


 \end{proof}



\subsection{Numerical Code}
The supplementary code for this project can be found at https://github.com/ck768/1D-Hydrogenic-Ion-Numerical.

\section*{Acknowledgement} The authors thank Monika Winklmeier, Eric Ling, and Lawrence Frolov for fruitful discussions, and for reading an earlier draft of this paper.
\bibliographystyle{plain}

\end{document}